\documentclass[12pt,a4paper]{amsart}

\usepackage[latin1]{inputenc}     
\usepackage{amsmath,amssymb,amsthm}
\usepackage{latexsym}
\usepackage{amsfonts}
\usepackage{bbm,bm,dsfont} 
\usepackage{eucal}
\usepackage{color} 
\usepackage{graphicx}
\usepackage{epstopdf}

\numberwithin{equation}{section} 

\theoremstyle{definition}
\newtheorem{proposition}{Proposition}
\newtheorem{definition}{Definition}

\newtheorem{remark}{Remark}
\newtheorem{theorem}{Theorem}
\newtheorem{lemma}{Lemma}

\newcommand{\mr}[1]{\mathrm{#1}}
\newcommand{\mc}[1]{\mathcal{#1}}

\newcommand{\ms}[1]{\mathsf{#1}}
\newcommand{\mb}[1]{\mathbb{#1}}

\newcommand{\sis}[2]{\langle#1|#2\rangle}

\newcommand{\tr}[1]{\mr{tr}[#1]}

\newcommand{\ptr}[2]{\mr{tr}_{#1}[#2]}
\newcommand{\dual}[2]{\langle #1,#2\rangle}
\newcommand{\la}{\langle}
\newcommand{\ra}{\rangle}
\newcommand{\ovl}[1]{\overline{#1}}
\newcommand{\hil}{\mathcal{H}}
\newcommand{\id}{\mathbbm{1}} 

\newcommand{\f}{\varphi}

\newcommand{\Om}{\Omega}
\newcommand{\om}{\omega}

\begin{document}

\title{Robustness of incompatibility for quantum devices}

\author{Erkka Haapasalo}
\email{ethaap@utu.fi}
\address{Turku Centre for Quantum Physics, Department of Physics and Astronomy, University of Turku, FI-20014 Turku, Finland}

\begin{abstract}

A robustness measure for incompatibility of quantum devices in the lines of the robustness of entanglement is proposed. The concept of general robustness measures is first introduced in general convex-geometric settings and these ideas are then applied to measure how incompatible a given pair of quantum devices is. The robustness of quantum incompatibility is calculated in three special cases: a pair of Fourier-coupled rank-1 sharp observables, a pair of decodable channels, where decodability means left-invertibility by a channel, and a pair consisting of a rank-1 sharp observable and a decodable channel.\\[7pt]

\newline
\noindent
{\bf PACS-numbers}: 03.65.-w, 03.65.Ta, 03.67.-a, 03.67.Mn
\end{abstract}
\maketitle


\section{Introduction}\label{intro}

As an inherently probabilistic construction, quantum theory abounds convex sets: the sets of states, observables, state changes, and measurements of a quantum system are all convex. Unlike in classical probability theories, these quantum theoretical convex structures are not simplexes, i.e.,\ states and measurements cannot be decomposed into combinations of extreme points in a unique way. This gives rise to many of the interesting aspects of quantum theory.

The rich structure of the set of quantum states has been extensively studied; see \cite{kiet} and references therein. Especially entanglement, as a truly quantum phenomenon, and its detection is the focus of great attention. In this paper, we concentrate on another peculiarity of quantum theory that has no counterpart in the classical world: incompatibility. Classical measurements can be carried out freely together and the measurements do not alter the system. On the quantum side, however, this no longer applies. There are many interesting pairs of quantum observables and measurements that do not allow any joint measurements or realizations. A canonical example is the position-momentum pair or any generalized Weyl pair.

In general, quantum incompatibility of a pair of quantum devices (observables, state-changes, instruments,\ldots) is defined as the impossibility of joining the devices into a single quantum device from which the original devices could be obtained by reduction. We give rigorous definitions for incompatibility in all the cases studied in this paper but an all-encompassing definition of quantum incompatibility can be found, e.g.,\ from \cite{HaHePe13}. It should be pointed out that the set of quantum states does not exhibit incompatibility; any pair of states can be joined into a bipartite state from which the original states can be obtained as partial traces.

Quantum incompatibility can be seen as a special resource like entanglement. That is, incompatibility is not simply a hindrance but rather a valuable non-classical feature that can be utilized in, e.g.,\ quantum information processing. In fact, there are connections between entanglement and incompatibility: it was recently shown in \cite{Brunner_etal2014, GuMoUo14} that incompatibility of quantum observables and EPR-steering of quantum states are operationally linked. Incompatibility as a resource is thus strongly related to the resource theory of steering. Moreover, a quantum channel is entanglement braking if and only if its transpose maps any observable pair into a jointly measurable (compatible) pair \cite{Pusey15}.

There are several measures for quantum entanglement one of which is the {\it robustness of entanglement} originally presented in \cite{VidalTarrach1998} that is purely based on the convex-geometric structure of the set of quantum states. Similar convexity based distance measures introduced for quantum convex sets include the {\it boundariness} defined in \cite{HaSeZi14} and the {\it steerable weight} introduced in \cite{steerableweight} quantifying the presence of EPR-steering. In this paper, we introduce a robustness measure for quantum incompatibility in the lines of robustness of entanglement. This quantity measures how well a given pair of quantum devices resists combining into a joint device under noise. Quantifying incompatibility of quantum observables has been earlier studied from a somewhat different viewpoint in \cite{Busch_etal2013, HeKiRe15, Heinosaari_etal2014}, but here we extend the notion of robustness of incompatibility to encompass all relevant quantum measurement device pairs.

A general description of robustness measures is given in Section \ref{general}. In Section \ref{sec:math}, we review the basic descriptions for the essential quantum apparati and, in Section \ref{sec:comp}, define the concept of compatibility of these apparati and introduce the robustness measures for incompatibility. In Section \ref{sec:robobs}, we study some special properties of the robustness of incompatibility. We calculate the robustness of incompatibility in three exemplary cases in Section \ref{sec:ex}.

\section{General robustness measures}\label{general}

The sets of measurement devices in any general statistical physical theory are naturally endowed with a convex structure. Namely, suppose that the set of devices under study is ${\bf Q}$ and $\Phi_1,\,\Phi_2\in{\bf Q}$ are devices of the same type. Then one can realize a device $\Phi\in{\bf Q}$ by applying $\Phi_1$ with probability $t\in[0,1]$ and by applying $\Phi_2$ with probability $1-t$. By expressing $\Phi$ as $t\Phi_1+(1-t)\Phi_2$, ${\bf Q}$ becomes a convex set. In a convex combination $t\Phi_1+(1-t)\Phi_2$, we view the coefficients $t$ as random noise or perturbation caused by statistical mixing of devices, and we use the term `noise' also in general convex geometries even in the absence of direct physical link. In what follows, we study convex sets $K$ of devices or device pairs (which are also naturally convex  by defining $t(\Phi_1,\Psi_1)+(1-t)(\Phi_2,\Psi_2)=(t\Phi_1+(1-t)\Phi_2,t\Psi_1+(1-t)\Psi_2)$) or, in general, {\it selection procedures} for a physical experiment with regards to a particular task or resource of the selections, such as entanglement (of individual states) or incompatibility (of device pairs) in the quantum case. We assume that there is a subset $L_0\subset K$ useless selections with respect to the property we are studying and, moreover, that this set is convex as well, i.e.,\ random mixtures of useless selections are also useless. In order to quantify the usefulness of an element $x\in K\setminus L_0$, we propose a way to measure the `distance' of elements $x\in K$ from $L_0$.

Let $V$ be a real vector space and $F\subset V$ an affine subspace. From now on, we fix a subset $L\subset F$ whose relative complement $F\setminus L=:L_0$ is convex; $L_0$ is to be thought of as the set of useless devices with respect to the task at hand. In the physical situations we will study later on, $F$ will be the minimal affine subspace generated by both $L$ and $L_0$, and $L_0$ will be {\it absorbing within $F$} in the sense that there is $y_0\in L_0$ such that for any $x\in F$ there is $t\in(0,1]$ such that $tx+(1-t)y_0\in L_0$. These assumptions are, however, not necessary in this section. We measure the distance of $x\in F$ to $L_0$ by quantifying the least amount of noise (from $L_0$ or some other convex subset of $F$) to be added to an element $x\in F$ in order to enter $L_0$, or, in other words, how robustly $x$ stays outside $L_0$ under added noise.

Let us make the following auxiliary definition.

\begin{definition}\label{def:relrob}
For any $x,\,y\in F$, let us define
$$
w_L(x|y)=\sup\{t\in[0,1]\,|\,tx+(1-t)y\in L_0\},
$$
where we define $\sup\emptyset=0$. We call $w_L(x|y)$ as the {\it relative $L$-robustness of $x$ relative to $y$}.
\end{definition}

It is immediate that, for a locally convex topological vector space $V$, whenever $y\in L_0$, we have $w_L(x|y)=1$ if and only if $x$ is in the closure of $L_0$. The number $1-w_L(x|y)$ is the least amount of noise in the form of a specific element $y$ from $F$ we have to add to $x$ in order to enter $L_0$, i.e.,\ the greatest additional noise in the form of $y$ the element $x$ tolerates without getting indiscriminable from $L_0$. It is immediate that, if $V$ is a locally convex space and $L_0$ is closed, then the supremum in the definition of $w_L(x|y)$ is attained when $y\in L_0$, i.e.,\ $w_L(x|y)x+(1-w_L(x|y))y\in L_0$. See Figure \ref{kuva1} for a visualization of the relative robustness. We can prove the following properties for the relative robustness:

\begin{proposition}\label{prop:relrob}
Let $x,\,y\in F$ be any fixed elements. Let us denote formally $1/t=\infty$ whenever $t=0$.
\begin{itemize}
\item[(a)] The function $w_L(\cdot|y)^{-1}:F\to\mb R\cup\{\infty\}$ is convex, i.e.,
$$
\frac1{w_L(tx_1+(1-t)x_2|y)}\leq\frac{t}{w_L(x_1|y)}+\frac{1-t}{w_L(x_2|y)}
$$
for any $x_1,\,x_2\in F$, $0\leq t\leq1$.
\item[(b)] The function $\big(1-w_L(x|\cdot)\big)^{-1}:L_0\to\mb R\cup\{\infty\}$ is concave, i.e.,
$$
\frac1{1-w_L(x|sy_1+(1-s)y_2)}\geq\frac{s}{1-w_L(x|y_1)}+\frac{1-s}{1-w_L(x|y_2)}
$$
for any $y_1,\,y_2\in L_0$, $0\leq s\leq1$.
\end{itemize}
\end{proposition}

\begin{proof}
Let us prove item (a). Pick $x_1,\,x_2\in F$, $y\in F$ and $t\in[0,1]$ and denote $x=tx_1+(1-t)x_2$. We may restrict to the case where $w_L(x_r|y)>0$, $r=1,\,2$. Choose $0<t_r<w_L(x_r|y)$ such that $z_r:=t_rx_r+(1-t_r)y\in L_0$, $r=1,\,2$. Through simple calculations, one finds that, denoting $t_0=\big(t/t_1+(1-t)/t_2\big)^{-1}$ and $s=tt_0/t_1$, one may write
$$
L_0\ni sz_1+(1-s)z_2=t_0x+(1-t_0)y.
$$
This means that $w_L(x|y)\geq t_0$ and, as one lets $t_r\uparrow w_L(x_r|y)$, $r=1,\,2$, the claim is proven.

We go on to proving item (b). Pick $y_1,\,y_2\in L_0$ and $s\in[0,1]$. If $w_L(x|y_1)=0$ (or $w_L(x|y_2)=0$), then set $t_1=0$ ($t_2=0$), otherwise, set $t_1\in[0,w_L(x|y_1))$ and $t_2\in[0,w_L(x|y_2))$. Denote $y_0=sy_1+(1-s)y_2$ and define $t_0\in[0,1]$ through $1-t_0=\big(s/(1-t_1)+(1-s)/(1-t_2)\big)^{-1}$. (Note that $t_1,\,t_2<1$.) Let $z_1\in L_0$, $z_2\in L_0$ be such that
$$
z_1=t_1x+(1-t_1)y_1\in L_0,\qquad z_2=t_2x+(1-t_2)y_2\in L_0.
$$
Direct calculation shows that we may write $z=t_0x+(1-t_0)y$ where
$$
z=rz_1+(1-r)z_2,\quad r=\frac{s(1-t_2)}{s(1-t_2)+(1-s)(1-t_1)},
$$
and, hence, $z\in L_0$ so that, by definition, $t_0\leq w_L(x|y)$. This amounts to
$$
\frac{1}{1-w_L(x|y)}\geq\frac{s}{1-t_1}+\frac{1-s}{1-t_2},
$$
and as we let $t_1\uparrow w_L(x|y_1)$ and $t_2\uparrow w_L(x|y_2)$, we obtain the desired result.
\end{proof}

\begin{center}
\begin{figure}
\includegraphics[scale=0.1]{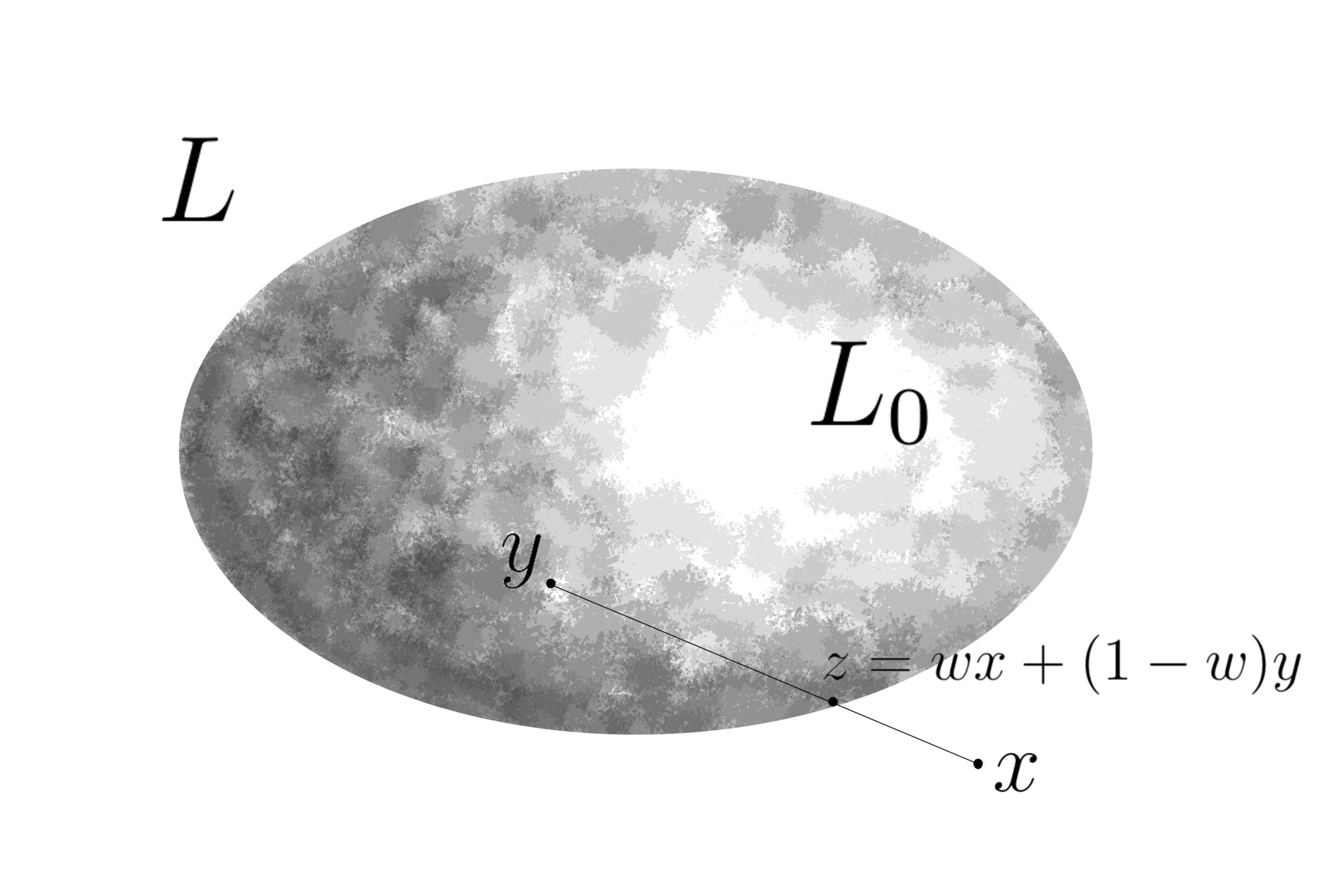}
\caption{\label{kuva1} Illustration of the relative robustness function. The points $x\in F$ and $y\in L_0$ are given and typically $x\in L$. The relative robustness is obtained by considering the line segment joining $x$ and $y$ and quantifying the weight $w$ with which $x$ appears in the boundary element $z$; in this case $w=w_L(x|y)$. One has $z\in L_0$ if $L_0$ is closed, i.e.,\ the supremum in the evaluation of the absolute robustness is actually obtained.}
\end{figure}
\end{center}

\begin{definition}\label{def:robust}
For any $x\in F$ let us define
$$
w_L(x)=\sup_{y\in L_0}w_L(x|y).
$$
We call $w_L(x)$ as the {\it (absolute) $L$-robustness of $x$}.
\end{definition}

The absolute robustness $w_L(x)$ measures the overall `distance' of $x$ to $L_0=F\setminus L$ in the sense that, whenever $L_0$ is a compact subset of a locally convex space $V$, $x\in L_0$ if and only if $w_L(x)=1$. Moreover, $1-w_L(x)$ is the greatest amount of random noise from $L_0$ that $x$ tolerates without being immersed in $L_0$. If $L_0$ is absorbing in $F$, then $w_L(x)>0$ for all $x\in F$.

Suppose that $x\in F$, $y,\,z\in L_0$ and $t\in[0,1]$ are such that $z=tx+(1-t)y$. Let us assume that there is an element $y'$ of $L_0$ on the line connecting $x,\,y$ and $z$ such that $y'=(1+p)y-px$ for some $p>0$, i.e.,\ $y'$ is `behind' $y$ when looked from $x$. It follows that $z=t'x+(1-t')y'$, where $t'=(t+p)/(1+p)=t+(1-t)p/(1+p)>t$. From this we may conclude that in order to increase $w_L(x|y)$ for a fixed $x\in F$ over $y\in L_0$ we should look for points $y$ of $L_0$ `furthest away' from $x$, as is also illustrated in Figure \ref{kuva2}. However, the maximizing point (if it exists) is typically not an extreme point of $L_0$.

We may prove the following:

\begin{proposition}\label{prop:robust}
The function $\big(w_L(\cdot)\big)^{-1}:F\to\mb R\cup\{\infty\}$ is convex, i.e.,\ for all $x_1,\,x_2\in F$ and $t\in[0,1]$
$$
\frac1{w_L(tx_1+(1-t)x_2)}\leq\frac{t}{w_L(x_1)}+\frac{1-t}{w_L(x_2)}.
$$
\end{proposition}

\begin{proof}
Pick $x_1,\,x_2\in F$ and $t\in[0,1]$ and define $x=tx_1+(1-t)x_2$. Clearly, we may restrict to the case where $w_L(x_1),\,w_L(x_2)>0$. Moreover, let $t_1\in(0,w_L(x_1))$, $t_2\in(0,w_L(x_2))$ and $y_1,\,y_2\in L_0$ be such that
$$
z_1=t_1x_1+(1-t_1)y_1\in L_0,\qquad z_2=t_2x_2+(1-t_2)y_2\in L_0.
$$
Define $t_0=\big(t/t_1+(1-t)/t_2\big)^{-1}$, and
\begin{eqnarray*}
y&=&\frac{t_0}{1-t_0}\Big(t\frac{1-t_1}{t_1}y_1+(1-t)\frac{1-t_2}{t_2}y_2\Big),\\
z&=&t_0\Big(\frac{t}{t_1}z_1+\frac{1-t}{t_2}z_2\Big).
\end{eqnarray*}
It is easy to check that $y$ is a convex combination of $y_1$ and $y_2$ and $z$ is a convex combination of $z_1$ and $z_2$, meaning $y\in L_0$ and $z\in L_0$. Furthermore, one may write $t_0x+(1-t_0)y=z$, so that $t_0\leq w_L(x)$ which is equivalent to
$$
\frac1{w_L(x)}\leq\frac{t}{t_1}+\frac{1-t}{t_2}.
$$
Letting $t_1\uparrow w_L(x_1)$ and $t_2\uparrow w_L(x_2)$, the claim is proven.
\end{proof}

\begin{center}
\begin{figure}
\includegraphics[scale=0.1]{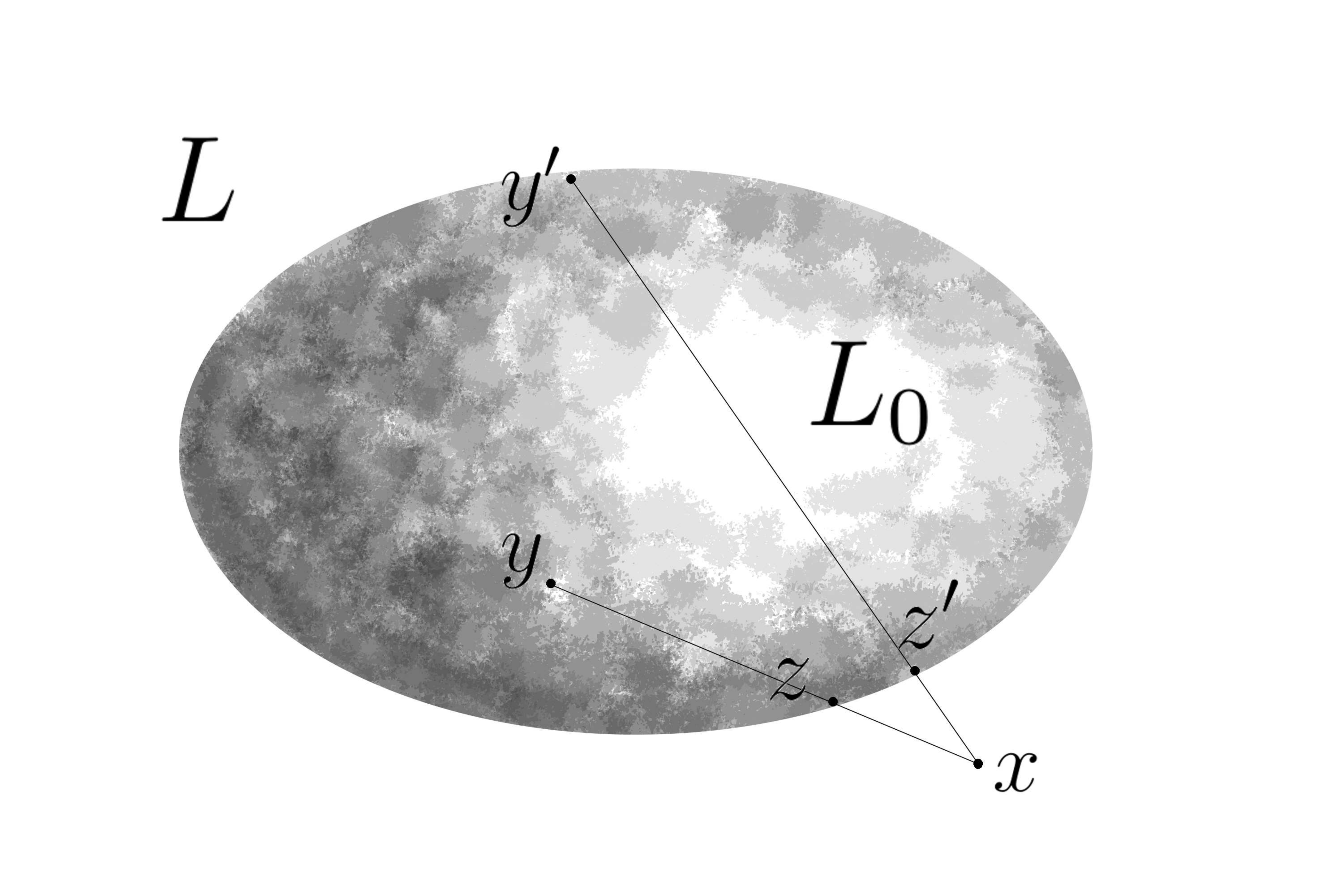}
\caption{\label{kuva2} Consider, as is illustrated, points $x\in L$ and $y,\,y'\in L_0$ situated so that $y'$ is fractionally further away from $x$ than $y$ is in the sense that, when we consider the boundary points $z$ and $z'$ found similarly as in Figure \ref{kuva1}, $x$ has a higher weight in $z'$ than in $z$, i.e.,\ $w(x|y)<w(x|y')$. It follows, as can be seen in the illustration, that these further-away points are located towards the boundary of the `other end' of $L_0$ seen from $x$.}
\end{figure}
\end{center}

In the physical situations, the set of actual selections cannot be described as a convex plane $F$ but rather as a restricted convex subset of a plane, and we are hence interested in distinguishing the useful selections out of the useless selections $y\in L_0$ within restricted sets. Thus, let us fix another convex set $K$, $L_0\subset K\subset F$. In most practical cases, the set $K$ will be compact with respect to some locally convex topology of $V$. Let us restrict the study of the earlier defined relative robustness measures to inputs from $K$ instead of the entire affine subspace $F$. We may define the following measure for an element of $K$ not being in $L_0$:

\begin{definition}\label{def:K-Lrobust}
We define the function $w_L^K:K\to\mb R$,
$$
w_L^K(x)=\sup_{y\in K}w_L(x|y),\qquad x\in K,
$$
and we call the number $w_L^K(x)$ as the {\it (absolute) $(K,L)$-robustness of $x$}.
\end{definition}

\begin{center}
\begin{figure}
\includegraphics[scale=0.1]{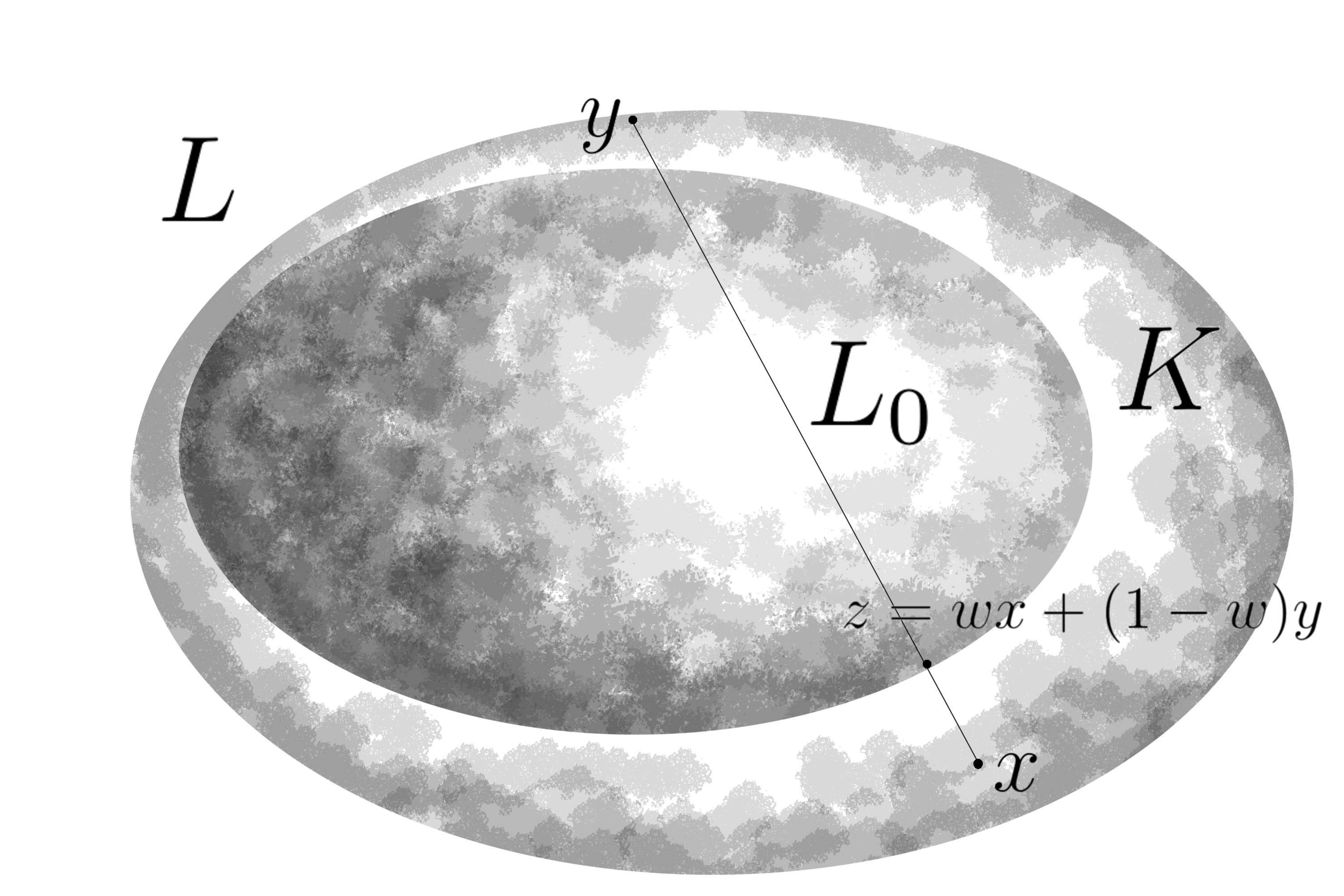}
\caption{\label{kuva3} The absolute $(K,L)$-robustness of a point $x\in K\setminus L_0$ (i.e.,\ $x$ is in the outer layer of the highlighted area in the illustration) is obtained by considering the points $y$ of $K$ being fractionally as far away from $x$ as possible in sense clarified in Figure \ref{kuva2}. However, now we do not require $y$ to be in $L_0$. Clearly, the optimizing $y$ is on the boundary of $K$.}
\end{figure}
\end{center}

The $(K,L)$-robustness $w_L^K$ measures `distances' of elements in $K$ from $L_0$ much in the same way as $w_L$; especially, in the case where $L_0$ and $K$ are compact, $x\in K$ is in $L_0$ if and only if $w_L^K(x)=1$. Again, $1-w_L^K(x)$ is the greatest amount of noise from $K$ that $x$ tolerates without being immersed in $L_0$. For geometric insight into $w_L^K$ and for how $w_L^K$ differs from $w_L$, see Figure \ref{kuva3}. Like $1/w_L$ also $1/w_L^K$ is convex; the proof is essentially the same as the proof for $w_L$. The $(K,L)$-robustness is typically easier to compute and, hence, we mainly concentrate on this measure in our examples in Section \ref{sec:ex}.

\begin{remark}\label{rem:R}
There is another, largely equivalent way to define the robustness measures. For any $x,\,y\in F$, let us define $R_L(x|y)$ as the infimum of the numbers $\lambda\in[0,\infty]$ such that $(1+\lambda)^{-1}(x+\lambda y)\in L_0$, where we set $\inf\emptyset:=\infty$. Hence, we can set up the functions $R_L:F\to\mb R\cup\{\infty\}$ and $R_L^K:K\to\mb R\cup\{\infty\}$,
$$
R_L(x)=\inf_{y\in L_0}R_L(x|y),\qquad R_L^K(x)=\inf_{y\in K}R_L(x|y).
$$
It is immediate that $R_L=1/w_L-1$ and $R_L^K=1/w_L^K-1$. Hence, according to Proposition \ref{prop:robust}, $R_L$ and, similarly, $R_L^K$ are convex. If $L_0$ (and, additionally, $K$) is compact and convex, it follows that $R_L(x)=0$ (resp. $R_L^K=0$) if and only if $x\in L_0$. The measures $R_L$ and $R_L^K$ are, thus, more reminiscent of true measures of distance from $L_0$ than $w_L$ and $w_L^K$, but because of the ease of calculation and clearer geometric meaning we use the measures $w_L$ and $w_L^K$, instead. Of course, all our results concerning the measures $w_L$ and $w_L^K$ are transferable into properties of $R_L$ and $R_L^K$.

The measures $R_L$ and $R_L^K$ have been applied in quantum physics to measure the entanglement of states of combined quantum systems. In this case, $K$ is the set of states of a combined (bipartite) quantum system and $L_0$ is the subset of separable quantum states, i.e.,\ the closed convex hull of product states $\rho_1\otimes\rho_2$. The measure $R_L$, denoted simply $R$ was, in this context, introduced in \cite{VidalTarrach1998} as the {\it robustness of entanglement}. In the sequel, we study the measures $w_L$ and $w_L^K$ in the context of quantum incompatibility.
\end{remark}

\section{The mathematical description of basic quantum devices}\label{sec:math}

In this section we fix complex and separable Hilbert spaces $\hil$ and $\mc K$; we denote the algebra of bounded operators on $\hil$ (respectively $\mc K$) by $\mc L(\hil)$ (respectively $\mc L(\mc K)$). The unit of $\mc L(\hil)$ (the identity operator) is represented by $\id_\hil$, although we usually omit the subscript if there is no danger of confusion. We denote the set of trace class-operators on $\hil$ (respectively on $\mc K$) by $\mc T(\hil)$ (respectively by $\mc T(\mc K)$); when endowed with the trace norm the set of trace-class operators becomes a Banach space. We define $\mc S(\hil)$ (respectively $\mc S(\mc K)$) as the set of positive elements of $\mc T(\hil)$ (respectively of $\mc T(\mc K)$) of trace 1.

Furthermore, $(\Om,\Sigma)$ and $(\Om',\Sigma')$ will be measurable spaces, i.e.,\ $\Om$ (respectively $\Om'$) is a non-empty set and $\Sigma$ (respectively $\Sigma'$) is a $\sigma$-algebra of subsets of $\Om$ (respectively $\Om'$). Additionally we assume that $(\Om,\Sigma)$ and $(\Om',\Sigma')$ are standard Borel (i.e.,\ $\sigma$-isomorphic to the Borel measurable space of a Borel subset of a Polish space); this makes the notion of post-processing introduced later easier to handle.

The quantum state space of a system described by the Hilbert space $\hil$ is identified with $\mc S(\hil)$. An $(\Om,\Sigma)$-valued quantum observable $\ms M$ is an affine map of $\mc S(\hil)$ into the set of probability measures on $(\Om,\Sigma)$, $\rho\mapsto p^{\ms M}_\rho$. The number $p^{\ms M}_\rho(X)$ is the probability of registering an outcome from the set $X\in\Sigma$ in a measurement of $\ms M$ when the system is in state $\rho$.

Hence, an observable is represented by (and, indeed, from now on identified with) a {\it normalized positive-operator-valued measure} (POVM), i.e.,\ an $(\Om,\Sigma)$-valued observable $\ms M$ on a physical system described by $\hil$ is a map $\ms M:\Sigma\to\mc L(\hil)$ such that for any state $\rho\in\mc S(\hil)$ the function $p^{\ms M}_\rho:\Sigma\to\mb R$, $p^{\ms M}_\rho(X)=\tr{\rho\ms M(X)}$, is a probability measure. This means that, as an operator-valued set function, $\ms M$ is weakly $\sigma$-additive, the range $\mr{ran}\,\ms M=\{\ms M(X)\,|\,X\in\Sigma\}$ consists of positive operators, and $\ms M(\Om)=\id_\hil$. A particular class of observables, the set of {\it sharp observables}, is made up of {\it projection-valued measures} (PVMs) $\ms P$ whose range consists solely of projections. We denote the set of $(\Om,\Sigma)$-valued observables of a system described by $\hil$ (identified with POVMs) by ${\bf Obs}(\Sigma,\hil)$.

A transformation of a system associated with the Hilbert space $\hil$ into a system associated with a possibly different Hilbert space $\mc K$ is described by an affine map $\mc E:\mc S(\hil)\to\mc S(\mc K)$ that is completely positive. This means that, for the transposed map $\mc E^*:\mc L(\mc K)\to\mc L(\hil)$,
$$
\sum_{j,k=1}^n\sis{\f_j}{\mc E^*(B_j^*B_k)\f_k}\geq0
$$
for all $n=1,\,2,\ldots$, $\f_1,\ldots,\,\f_n\in\hil$, and $B_1,\ldots,\,B_n\in\mc L(\mc K)$, and $\mc E^*(\id_{\mc K})=\id_\hil$. Such maps are called as {\it channels} and we denote the set of channels $\mc E:\mc S(\hil)\to\mc S(\mc K)$ by ${\bf Ch}(\hil,\mc K)$.

Let $\mc K$ be another Hilbert space, $\ms M\in{\bf Obs}(\Sigma,\hil)$, and $\ms N\in{\bf Obs}(\Sigma,\mc K)$. We say that $\ms N$ is {\it pre-processing of $\ms M$ (by a channel $\mc E\in{\bf Ch}(\hil,\mc K)$)} if $\ms N=\mc E^*\circ\ms M$. In other words, $p^{\ms N}_\rho=p^{\ms M}_{\mc E(\rho)}$. Thus a measurement of $\ms N$ can be implemented by first transforming the system with the channel $\mc E$ and then measuring $\ms M$ on the transformed system.

Assume now that $\ms M\in{\bf Obs}(\Sigma,\hil)$ and $\ms N\in{\bf Obs}(\Sigma',\hil)$. Suppose that $\rho_0$ is a fixed faithful state operator on $\hil$ and denote $\mu=p^{\ms M}_{\rho_0}$. If there exists a function $\beta:\Sigma'\times\Om\to\mb R$ such that
\begin{itemize}
\item[(i)] for any $Y\in\Sigma'$ the function $\beta(Y,\cdot):\Om\to\mb R$ is $\Sigma$-measurable and
\item[(ii)] the set function $\beta(\cdot,\om):\Sigma'\to\mb R$ is a probability measure for $\mu$- almost all $\om\in\Om$
\end{itemize}
(i.e.,\ $\beta$ is a Markov kernel) such that
$$
\ms N(Y)=\int_\Om\beta(Y,\om)\ms M(d\om)
$$
for all $Y\in\Sigma'$, we say that $\ms N$ is a {\it post-processing of $\ms M$ (with the Markov kernel $\beta$)}. We usually write, in this context, $\ms N=\ms M^\beta$. Hence, we may measure $\ms N$ by first measuring $\ms M$ and then processing the outcome probability distribution of $\ms M$ by $\beta$. For more on post-processing (or coarse-graining) especially in finite-outcome settings, see, \cite{MaMu90}. The appropriate generalizations needed in the case involving measurable spaces that are not standard Borel and deeper issues in post-processing are particularly studied in \cite{JePu07, JePuVi08}.

We say that a map $\Phi:\mc T(\hil)\to\mc T(\mc K)$ is an {\it operation} if it is linear, completely positive, and $\tr{\Phi(\rho)}\leq1$ for all $\rho\in\mc S(\hil)$. It follows that an operation is trace-norm continuous. Complete positivity of $\Phi$ means that the (normal) dual map $\Phi^*:\mc L(\mc K)\to\mc L(\hil)$ is completely positive. We call a weakly $\sigma$-additive map $\Sigma\ni X\mapsto\Gamma_X$, where $\Gamma_X$ is an operation for all $X\in\Sigma$ and $\Gamma_\Om\in{\bf Ch}(\hil,\mc K)$, an {\it instrument}. Weak $\sigma$-additivity means that, for any $B\in\mc L(\mc K)$ and $T\in\mc T(\hil)$, the map $X\mapsto\tr{B\Gamma_X(T)}$ is $\sigma$-additive. We denote the set of instruments associated with the measurable space $(\Om,\Sigma)$ and the Hilbert spaces $\hil$ and $\mc K$ as above by ${\bf Ins}(\Sigma,\hil,\mc K)$.

An instrument $\Gamma\in{\bf Ins}(\Sigma,\hil,\mc K)$ is a mathematical description of a measurement process; when the system is in the state $\rho$, $\Gamma_X(\rho)$ is the non-normalized conditional state after the measurement described by $\Gamma$ conditioned by an outcome being measured in the subset $X$ and $\tr{\Gamma_X(\rho)}$ is the probability of registering an outcome from $X$ when the input state of the system is $\rho$. Hence, an instrument combines the description of measurement outcome statistics depending on the input state of the system (observable) with the knowledge of the conditional state changes (with the unconditioned state-change $\Gamma_\Om$ being a channel).

The sets ${\bf Obs}(\Sigma,\hil)$, ${\bf Ch}(\hil,\mc K)$, and ${\bf Ins}(\Sigma,\hil)$ of observables, channels, and instruments are convex, as they should as measurement device sets of a statistical physical theory. As an example,\ for $\Gamma,\,\Gamma'\in{\bf Ins}(\Sigma,\hil)$ and $t\in[0,1]$, the convex combination $t\Gamma+(1-t)\Gamma'\in{\bf Ins}(\Sigma,\hil)$ is defined through $\big(t\Gamma+(1-t)\Gamma'\big)_X(\rho)=t\Gamma_X(\rho)+(1-t)\Gamma'_X(\rho)$ for all $X\in\Sigma$ and $\rho\in\mc S(\hil)$. It is noteworthy that, unlike in classical physical theories, the convex structures in quantum theory allow typically (uncountably) many decompositions into extreme points for states \cite{cassinelli} as well as for measurement devices meaning that a mixture $t\Phi+(1-t)\Psi$ of quantum devices $\Phi$ and $\Psi$ with $0<t<1$ cannot be considered as an ensemble of devices where the device realized would be, in fact, $\Phi$ with probability $t$ and $\Psi$ with probability $1-t$.

\section{Quantum compatibility and incompatibility}\label{sec:comp}

In this section we give a description of compatibility and the complementary notion of incompatibility of the relevant quantum devices, namely, observables and channels. Incompatibility of quantum observables is a well-known issue; see, e.g.,\ review of the subject in \cite{pekka03} and references therein. The compatibility of other types of devices is dealt with earlier, e.g.,\ in \cite{HaHePe13, HeMiRe14}, and our discussion here follows the definitions made in these references. Let us fix Hilbert spaces $\hil$ and $\mc K$ and the standard Borel value spaces $(\Om,\Sigma)$ and $(\Om',\Sigma')$.

We say that observables $\ms M\in{\bf Obs}(\Sigma,\hil)$ and $\ms N\in{\bf Obs}(\Sigma',\hil)$ are {\it compatible} or {\it jointly measurable} if there is a third value space $(\overline{\Om},\overline{\Sigma})$ and an observable $\ms G\in{\bf Obs}(\overline{\Sigma},\hil)$ such that $\ms M$ and $\ms N$ are post-processings of $\ms G$. Joint measurability of $\ms M$ and $\ms N$ means that we may determine the outcome statistics of $\ms M$ and $\ms N$ from the outcome statistics of $\ms G$ by classical means (Markov kernels).\footnote{Here is also the connection of joint measurability of observables and non-existence of steering: steering is not present when one party (Alice) cannot steer the other party's (Bob) state outside the range obtained by post-processings of a fixed ensemble of states over a hidden variable with her measurements (observables). Steering will never be realized independently of the initial state if and only if the observables Alice measures are jointly measurable.} If a pair of observables is not jointly measurable, we say that the observables are {\it incompatible}.

Since the value spaces of the observables we study are assumed to be standard Borel, we find an observable $\ms G\in{\bf Obs}(\Sigma\otimes\Sigma',\hil)$ on the product measurable space $(\Om\times\Om',\Sigma\otimes\Sigma')$ for any pair $\ms M\in{\bf Obs}(\Sigma,\hil)$ and $\ms N\in{\bf Obs}(\Sigma',\hil)$ of jointly measurable observables such that
$$
\ms G(X\times\Om')=\ms M(X),\qquad\ms G(\Om\times Y)=\ms N(Y)
$$
for all $X\in\Sigma$ and $Y\in\Sigma'$. We call such an observable $\ms G$ as a {\it joint observable} for $\ms M$ and $\ms N$. The joint observable of a pair of jointly measurable observables need not be unique. However, if $\ms M$ is an extreme point of ${\bf Obs}(\Sigma,\hil)$ or $\ms N$ is an extreme point of ${\bf Obs}(\Sigma',\hil)$ and $\ms M$ and $\ms N$ are jointly measurable, their joint observable is unique \cite{HaHePe13}.

Note that any observable $\ms M\in{\bf Obs}(\Sigma,\hil)$ is compatible with itself. Indeed, let $(\mc M,\ms P,J)$ be a (not necessarily minimal) Na\u{\i}mark dilation for $\ms M$, where $\mc M$ is a Hilbert space, $\ms P:\Sigma\to\mc L(\mc M)$ is a projection valued measure, and $J:\hil\to\mc M$ is an isometry such that $\ms M(X)=J^*\ms P(X)J$ for all $X\in\Sigma$. The observable $\ms G\in{\bf Obs}(\Sigma\otimes\Sigma,\hil)$, $\ms G(X\times Y)=J^*\ms P(X\cap Y)J$, $X,\,Y\in\Sigma$, is clearly a joint observable for the pair $(\ms M,\ms M)$. We say that an observable $\ms T\in{\bf Obs}(\Sigma,\hil)$ is {\it trivial}, if there is a probability measure $p:\Sigma\to[0,1]$ such that $\ms T(X)=p(X)\id_\hil$ for all $X\in\Sigma$. It is immediate that the trivial observables are compatible with any other observables.

Whenever $\mc F\in{\bf Ch}(\hil,\mc K_1\otimes\mc K_2)$, we may define its {\it marginals} $\mc F_{(j)}\in{\bf Ch}(\hil,\mc K_j)$, $j=1,\,2$, through
$$
\mc F_{(1)}=\mr{tr}_{\mc K_2}[\mc F(\cdot)],\qquad\mc F_{(2)}=\mr{tr}_{\mc K_1}[\mc F(\cdot)]
$$
or, equivalently in the dual form,
$$
\mc F_{(1)}^*(A)=\mc F^*(A\otimes\id_{\mc K_2}),\qquad\mc F_{(2)}^*(B)=\mc F^*(\id_{\mc K_1}\otimes B)
$$
for all $A\in\mc L(\mc K_1)$ and $B\in\mc L(\mc K_2)$. Clearly, the marginals are channels as well. The marginal channels describe, e.g.,\ the reduced dynamics of subsystems.

We say that $\mc E_j\in{\bf Ch}(\hil,\mc K_j)$, $j=1,\,2$, are {\it compatible} if they are marginals of a channel $\mc F\in{\bf Ch}(\hil,\mc K_1\otimes\mc K_2)$, i.e.,\ $\mc E_1=\mc F_{(1)}$ and $\mc E_2=\mc F_{(2)}$. In this case, we call $\mc F$ as a {\it joint channel} of $\mc E_1$ and $\mc E_2$. If a pair of channels is not compatible, we say that they are {\it incompatible}. Compatibility of channels parallels the joint measurability of observables. As with joint measurability, the joint channel of a compatible pair need not be unique. However, a similar sufficient condition for uniqueness can be established as in the case of jointly measurable observables \cite{HaHePe13}. Due to the deeper non-commutativity of channels, a channel may not be compatible with itself. Indeed, the no-cloning principle can be stated in the form that the identity channel $\mr{id}$, $\mr{id}(B)=B$, is not compatible with itself. We say that a channel that is not compatible with itself is {\it self-incompatible}. In a sense, the reason for the fact that (practically) any observable is self-compatible whereas a channel is not is that the output of an observable is classical information that can be copied freely but the quantum output of a channel cannot be copied because of the self-incompatibility of the identity channel.

Finally, we come to the compatibility criterion for a pair of a quantum observable and a channel. Such a pair is compatible if they can be combined in a single measurement, i.e.,\ an instrument. An instrument $\Gamma\in{\bf Ins}(\Sigma,\hil,\mc K)$ has the observable marginal $\Gamma_{(1)}\in{\bf Obs}(\Sigma,\hil)$ and the channel marginal $\Gamma_{(2)}\in{\bf Ch}(\hil,\mc K)$ defined by
$$
\Gamma_{(1)}(X)=\Gamma_X^*(\id_{\mc K}),\qquad\Gamma_{(2)}(\rho)=\Gamma_\Om(\rho)
$$
for all $X\in\Sigma$ and all $\rho\in\mc S(\hil)$. The observable $\Gamma_{(1)}$ is the observable whose measurement is realized by the measurement process described by $\Gamma$ and $\Gamma_{(2)}$ is the unconditioned total state change induced by the measurement.

\begin{center}
\begin{figure}
\includegraphics[scale=0.07]{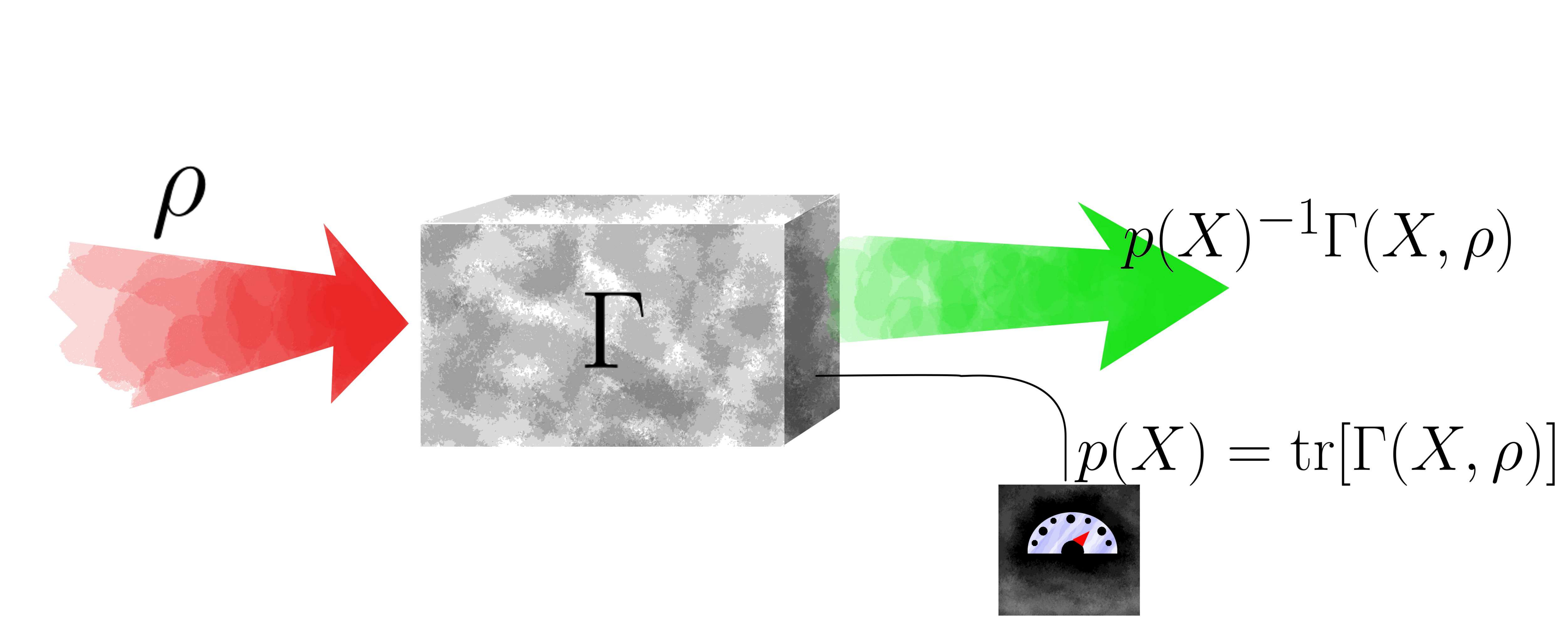}
\caption{\label{kuva4} Illustration of an instrument. The state entering the measurement device represented by the instrument $\Gamma$ is $\rho$. The instrument has the statistics arm (the lower branch right of the instrument box in the illustration) and the state change arm (the upper branch right of the instrument box). When a value is detected with certainty in the set $X$ in the statistics arm, which happens with probability $p(X)$, the state change arm gives the conditional state $\rho_X=p(X)^{-1}\Gamma_X(\rho)$. When the state changes are neglected, the statistics arm reduces to the observable $\Gamma_{(1)}$, and when the statistics are ignored, the state-change arm reduces to the channel $\Gamma_{(2)}$. All joint maps of compatible quantum devices can be visualized in the same way as input-output processors with multiple outputs. Each of the original compatible devices is obtained when all the other output arms are ignored except for the one associated with the particular device.}
\end{figure}
\end{center}

Compatibility of an observable $\ms M\in{\bf Obs}(\Sigma,\hil)$ and a channel $\mc E\in{\bf Ch}(\hil,\mc K)$ means that they can be combined in a single measurement, i.e.,\ there is an instrument $\Gamma\in{\bf Ins}(\Sigma,\hil,\mc K)$ such that $\ms M=\Gamma_{(1)}$ and $\mc E=\Gamma_{(2)}$, as is highlighted in Figure \ref{kuva4}. In this context, we call $\Gamma$ as a {\it joint instrument} for $\ms M$ and $\mc E$. Again, a compatible pair $(\ms M,\mc E)$ typically has infinitely many joint instruments but one sufficient condition for uniqueness of the joint instrument is the extremality of a marginal. A pair $(\ms M,\mc E)$ is defined to be {\it incompatible} if it is not compatible.

\subsection{Robustness of incompatibility}

We fix the sets ${\bf Q}_1$ and ${\bf Q}_2$ of quantum devices that are either observables or channels, i.e.,\ with fixed Hilbert spaces $\hil$, $\mc K$, and $\mc K'$ and value spaces $(\Om,\Sigma)$ and $(\Om',\Sigma')$, ${\bf Q}_1$ is either ${\bf Obs}(\Sigma,\hil)$ or ${\bf Ch}(\hil,\mc K)$ and ${\bf Q}_2$ is either ${\bf Obs}(\Sigma',\hil)$ or ${\bf Ch}(\hil,\mc K')$. We denote the set of compatible pairs within ${\bf Q}_1\times{\bf Q}_2$ by ${\bf Comp}$, i.e.,\ $\Phi_1\in{\bf Q}_1$ and $\Phi_2\in{\bf Q}_2$ are compatible if and only if $(\Phi_1,\Phi_2)\in{\bf Comp}$. The set ${\bf Q}_1\times{\bf Q}_2$ is endowed with a natural convex structure by defining for any $(\Phi_1,\Phi_2),\,(\Psi_1,\Psi_2)\in{\bf Q}_1\times{\bf Q}_2$ and $t\in[0,1]$ the combination $t(\Phi_1,\Phi_2)+(1-t)(\Psi_1,\Psi_2)=\big(t\Phi_1+(1-t)\Psi_1,t\Phi_2+(1-t)\Psi_2\big)$. Whatever the sets of devices involved, the reader may easily check that ${\bf Comp}$ is a convex subset of ${\bf Q}_1\times{\bf Q}_2$.

Denote by $L$ the relative complement of ${\bf Comp}$ with respect to the minimal affine subspace containing ${\bf Comp}$, which coincides with the minimal affine subspace containing ${\bf Q}_1\times{\bf Q}_2$; see Remark \ref{rem:yli1/2} for a proof of this fact. The product ${\bf Q}_1\times{\bf Q}_2$ we denote by $K$. We may define the robustness measures $w_L(\cdot|\cdot)$, $w_L$, and $w_L^K$ introduced in general form in Section \ref{general} for the set of compatible pairs ${\bf Comp}$. For simplicity, we denote $w_L^K(\cdot|\cdot)=:w(\cdot|\cdot)$, $w_L=:w$, and $w_L^K=:W$. Moreover, for any $\Phi_1,\,\Psi_1\in{\bf Q}_1$ and $\Phi_2,\,\Psi_2\in{\bf Q}_2$, we simplify our notations:
\begin{eqnarray*}
w\big((\Phi_1,\Phi_2)|(\Psi_1,\Psi_2)\big)&=:&w(\Phi_1,\Phi_2|\Psi_1,\Psi_2),\\
w\big((\Phi_1,\Phi_2)\big)&=:&w(\Phi_1,\Phi_2),\\
W\big((\Phi_1,\Phi_2)\big)&=:&W(\Phi_1,\Phi_2).
\end{eqnarray*}
When ${\bf Q}_1={\bf Q}_2={\bf Ch}(\hil,\mc K)$, we denote $W(\mc E,\mc E)=W(\mc E)$; this quantity we call as the {\it robustness of self-incompatibility of $\mc E$}.

In the case where ${\bf Q}_1={\bf Obs}(\Sigma,\hil)$ and ${\bf Q}_2={\bf Ch}(\hil,\hil)$, the quantity ${\bf Obs}(\Sigma,\hil)\ni\ms M\mapsto W(\ms M,\mr{id})$, where $\mr{id}$ stands for the identity channel in ${\bf Ch}(\hil,\hil)$, measures how well an approximate version of $\ms M$ can be measured while disturbing the system as little as possible. Equivalently, one may think of the number $W(\ms M,\mr{id})$ as the measure of how disturbing any measurement of the observable $\ms M$ inherently is. In Section \ref{sec:rank1id}, this quantity is calculated in the case of finite-dimensional $\hil$ for any rank-1 sharp observable.

\begin{remark}\label{rem:yli1/2}
In \cite{Busch_etal2013}, it was shown that, whenever $t\leq1/2$, $(\ms A,\ms B)\in{\bf Obs}(\Sigma,\hil)\times{\bf Obs}(\Sigma',\hil)$, and $(\ms S,\ms T)\in{\bf Obs}(\Sigma,\hil)\times{\bf Obs}(\Sigma',\hil)$ is a pair of trivial observables, then $t(\ms A,\ms B)+(1-t)(\ms S,\ms T)\in{\bf JM}(\Sigma,\Sigma',\hil)$, so that
$$
W(\ms A,\ms B)\geq w(\ms A,\ms B|\ms S,\ms T)\geq\frac12.
$$
Hence, $1/2$ is a global lower bound for the robustness $W$. Indeed, the same reasoning as that in \cite{Busch_etal2013} also applies to the compatibility questions involving other devices than observables as well as to the case of continuous observables not discussed in \cite{Busch_etal2013}. Let us study, e.g.,\ the case of observable-channel pairs, i.e.,\ ${\bf Q}_1={\bf Obs}(\Sigma,\hil)$ and ${\bf Q}_2={\bf Ch}(\hil,\mc K)$. Pick any probability measure $p$ on $(\Om,\Sigma)$ and any state $\sigma\in\mc S(\mc K)$ defining the associated trivial observable $\ms T_p=p(\cdot)\id_\hil$ and the constant channel $\mc T_\sigma:\rho\mapsto\sigma$. With any $\ms M\in{\bf Obs}(\Sigma,\hil)$, $\mc E\in{\bf Ch}(\hil,\mc K)$, and $t\in[0,1]$, one may set up the instrument $\Gamma\in{\bf Ins}(\Sigma,\hil,\mc K)$,
$$
\Gamma_X(\rho)=t\tr{\rho\ms M(X)}\sigma+(1-t)p(X)\mc E(\rho),\qquad\rho\in\mc S(\hil).
$$
The observable marginal of $\Gamma$ is easily found to be $t\ms M+(1-t)\ms T_p$ and the channel marginal is $t\mc T_\sigma+(1-t)\mc E$. Especially it follows that fixing any pair $(\ms T_p,\mc T_\sigma)$ like that above, one has that $\frac12(\ms M+\ms T_p)$ and $\frac12(\mc E+\mc T_\sigma)$ are compatible. Similar result is easily proven for channel-channel pairs. It thus follows that there is $x_0\in{\bf Comp}$ such that $\frac12(x+x_0)\in{\bf Comp}$ for all $x\in{\bf Q}_1\times{\bf Q}_2$ implying $W\geq1/2$ for any device pair

It also follows that the minimal affine subspace $F'$ containing ${\bf Comp}$ coincides with the minimal affine subspace $F$ containing ${\bf Q}_1\times{\bf Q}_2$. Indeed, trivially, $F'\subset F$. For the reversed inclusion, let us pick any $x_0\subset{\bf Comp}$ such that $\frac12(x+x_0)\in{\bf Comp}$ for all $x\in{\bf Q}_1\times{\bf Q}_2$. Let $z\in F$ meaning that there are $x,\,y\in{\bf Q}_1\times{\bf Q}_2$ and $\lambda\geq0$ such that $z=x_0+\lambda(x-y)$. Defining $x':=\frac12(x+x_0)\in{\bf Comp}$ and $y':=\frac12(y+x_0)\in{\bf Comp}$, it follows $z=x_0+2\lambda(x'-y')\in F'$ implying $F\subset F'$.

Following \cite{Heinosaari_etal2014}, in the case for observables operating in a $d$-dimensional ($d<\infty$) Hilbert space one can give an even tighter bound for robustness of incompatibility for observable pairs:
$$
W(\ms A,\ms B)\geq\frac{2+d}{2(1+d)}.
$$
One easily sees, using similar techniques as in \cite{Heinosaari_etal2014} that the above inequality holds also for the robustness measures $W$ involving quantum devices other than observables.
\end{remark}

\subsection{Ordering properties}\label{sec:robobs}

In this subsection we discuss some of the special features of the robustness measures for incompatibility. We will find out that the robustness measure behaves monotonically under certain orderings of device pairs. In the sequel, whenever $(\Phi,\Psi)$ is a device pair and we write $W(\Phi,\Psi)$ (or $w(\Phi,\Psi)$) we implicitly assume that the base set ${\bf Q}_1\times{\bf Q}_2\ni(\Phi,\Psi)$ contains compatible pairs, i.e.,\ if $\Phi$ and $\Psi$ are observables, they operate in the same Hilbert space, if $\Phi$ and $\Psi$ are channels, their input spaces coincides, and if $\Phi$ is an observable and $\Psi$ is a channel, then $\Phi$ operates in the input space of the channel $\Psi$.

Let us fix the Hilbert spaces $\hil$ and $\mc K$ and four standard Borel measurable spaces $(\Om,\Sigma)$, $(\Om',\Sigma')$, $(\tilde\Om,\tilde\Sigma)$, and $(\tilde\Om',\tilde\Sigma')$. We denote the subset of those pairs $(\ms M,\ms N)\in{\bf Obs}(\Sigma,\hil)\times{\bf Obs}(\Sigma',\hil)$ that are jointly measurable by ${\bf JM}(\Sigma,\Sigma',\hil)$. We define the corresponding sets of compatible observable pairs also for other pairs of value spaces.

Suppose that $\ms M\in{\bf Obs}(\Sigma,\mc K)$ is a pre-processing of $\ms A\in{\bf Obs}(\Sigma,\hil)$ and $\ms N\in{\bf Obs}(\Sigma',\mc K)$ is a pre-processing of $\ms B\in{\bf Obs}(\Sigma',\hil)$, both with the {\it same} channel $\mc E\in{\bf Ch}(\hil,\mc K)$, i.e.,\ $\ms M=\mc E^*\circ\ms A$ and $\ms N=\mc E^*\circ\ms B$. In this case, we write $(\ms M,\ms N)\leq_\mr{prae}(\ms A,\ms B)$. The relation $\leq_\mr{prae}$ is a pre-order in the set of observable pairs. We denote $(\ms A,\ms B)=_\mr{prae}(\ms M,\ms N)$ if $(\ms M,\ms N)\leq_\mr{prae}(\ms A,\ms B)$ and $(\ms A,\ms B)\leq_\mr{prae}(\ms M,\ms N)$ and say that the pairs $(\ms A,\ms B)$ and $(\ms M,\ms N)$ are {\it pre-processing equivalent}.

Another pre-order in the set of observable pairs is defined by post-processing: If $\tilde{\ms M}\in{\bf Obs}(\tilde\Sigma,\hil)$ is a post-processing of $\ms A\in{\bf Obs}(\Sigma,\hil)$ and $\tilde{\ms N}\in{\bf Obs}(\tilde\Sigma',\hil)$ is a post-processing of $\ms B\in{\bf Obs}(\Sigma',\hil)$ with possibly different Markov kernels, we write $(\tilde{\ms M},\tilde{\ms N})\leq_\mr{post}(\ms A,\ms B)$. We define the {\it post-processing equivalence} $=_\mr{post}$ for pairs of observables in the same way as the pre-processing equivalence.

Suppose that $(\ms A,\ms B)\in{\bf JM}(\Sigma,\Sigma',\hil)$ have a joint observable $\ms G$. Then $\mc E^*\circ\ms G$ is a joint observable for $(\mc E^*\circ\ms A,\mc E^*\circ\ms B)$ for any channel $\mc E$. This means that, if the pair $(\ms A,\ms B)$ is jointly measurable and $(\tilde{\ms M},\tilde{\ms N})\leq_\mr{prae}(\ms A,\ms B)$, then $(\tilde{\ms M},\tilde{\ms N})$ is jointly measurable. Moreover, it follows immediately from the definition of joint measurability that, if $(\ms A,\ms B)$ is jointly measurable and $(\ms M,\ms N)\leq_\mr{post}(\ms A,\ms B)$, then $(\ms M,\ms N)$ is jointly measurable.


We may show that the $W$-robustness measure of incompatibility has the following properties. With a slight modification of the proof given here, one easily shows that these results also hold for $w$.

\begin{theorem}\label{theor:properties}
Let $(\ms A,\ms B)$, $(\ms M,\ms N)$, and $(\tilde{\ms M},\tilde{\ms N})$ be observable pairs.
\begin{itemize}
\item[(a)] If $(\ms M,\ms N)\leq_\mr{prae}(\ms A,\ms B)$, then $W(\ms M,\ms N)\geq W(\ms A,\ms B)$, and if $(\ms M,\ms N)=_\mr{prae}(\ms A,\ms B)$, then $W(\ms M,\ms N)=W(\ms A,\ms B)$.
\item[(b)] If $(\tilde{\ms M},\tilde{\ms N})\leq_\mr{post}(\ms A,\ms B)$, then $W(\tilde{\ms M},\tilde{\ms N})\geq W(\ms A,\ms B)$, and if $(\tilde{\ms M},\tilde{\ms N})=_\mr{post}(\ms A,\ms B)$, then $W(\tilde{\ms M},\tilde{\ms N})=W(\ms A,\ms B)$.
\end{itemize}
\end{theorem}

\begin{proof}
Clearly, if $W(\ms A,\ms B)=0$, the first claim in item (a) needs no proof. Assume, hence, that $W(\ms A,\ms B)>0$; in fact, according to Remark \ref{rem:yli1/2}, the robustness measure is always bounded from below by $1/2$. Suppose that $(\ms M,\ms N)\leq_\mr{prae}(\ms A,\ms B)$, the pre-processing carried out by a channel $\mc E\in{\bf Ch}(\hil,\mc K)$ and $t<W(\ms A,\ms B)$, and let $(\ms A_1,\ms B_1)\in{\bf Obs}(\Sigma,\hil)\times{\bf Obs}(\Sigma',\hil)$ and $(\ms A_2,\ms B_2)\in{\bf JM}(\Sigma,\Sigma',\hil)$ be such that
$$
t(\ms A,\ms B)+(1-t)(\ms A_1,\ms B_1)=(\ms A_2,\ms B_2).
$$
Let now $\ms M_r=\mc E^*\circ\ms A_r$ and $\ms N_r=\mc E^*\circ\ms B_r$ for $r=1,\,2$. Hence, $(\ms M_2,\ms N_2)$ is a compatible pair. It follows immediately (as one may check) that one may write
$$
t(\ms M,\ms N)+(1-t)(\ms M_1,\ms N_1)=(\ms M_2,\ms N_2),
$$
from which the first claim of item (a) follows as one lets $t\uparrow W(\ms A,\ms B)$. The second claim of item (a) follows from symmetry.

In the proof of item (b), we may again restrict to the case where $W(\ms A,\ms B)>0$. Assume now that $(\tilde{\ms M},\tilde{\ms N})\leq_\mr{post}(\ms A,\ms B)$, so that there are Markov kernels $\beta$ and $\gamma$ such that $\tilde{\ms M}=\ms A^\beta$ and $\tilde{\ms N}=\ms B^\gamma$, and $t\leq W(\ms A,\ms B)$, and let $(\ms A_1,\ms B_1)\in{\bf Obs}(\Sigma,\hil)\times{\bf Obs}(\Sigma',\hil)$ and $(\ms A_2,\ms B_2)\in{\bf JM}(\Sigma,\Sigma',\hil)$ be as above. Now, $(\ms A_2^\beta,\ms B_2^\gamma)$ is a compatible pair, and one may write
$$
t(\tilde{\ms M},\tilde{\ms N})+(1-t)(\ms A_1^\beta,\ms B_1^\gamma)=(\ms A_2^\beta,\ms B_2^\gamma)
$$
proving the first claim of item (b) as one lets $t\uparrow W(\ms A,\ms B)$. The second claim of item (b) is proven by symmetry again.
\end{proof}

Especially, if the observables $\ms A\in{\bf Obs}(\Sigma,\hil)$ and $\ms B\in{\bf Obs}(\Sigma',\hil)$ are post-processing maximal (i.e.,\ rank-1 observables), then the pair $(\ms A,\ms B)$ minimizes the robustness measure $W$, i.e., a pair of post-processing maximal observables require the greatest amount of noise to be added in order to be rendered jointly measurable.

Let now $\hil,\,\hil',\,\mc K_1,\,\mc K_2,\,\mc K_1'$, and $\mc K_2'$ be Hilbert spaces. We denote, e.g.,\ for the spaces $\hil$, $\mc K_1$, and $\mc K_2$, the set of compatible pairs in ${\bf Ch}(\hil,\mc K_1)\times{\bf Ch}(\hil,\mc K_2)$ by ${\bf Comp}(\hil,\mc K_1,\mc K_2)$.

For channels $\mc E'\in{\bf Ch}(\hil',\mc K_1)$, $\mc F'\in{\bf Ch}(\hil',\mc K_2)$, $\mc E\in{\bf Ch}(\hil,\mc K_1)$, and $\mc F\in{\bf Ch}(\hil,\mc K_2)$, we denote $(\mc E',\mc F')\leq_\mr{prae}(\mc E,\mc F)$ if there is a channel $\mc G\in{\bf Ch}(\hil',\hil)$ such that $\mc E'=\mc E\circ\mc G$ and $\mc F'=\mc F\circ\mc G$. This gives rise to the partial order $\leq_\mr{prae}$ associated to pre-processing of channel pairs. We denote by $=_\mr{prae}$ the corresponding equivalence relation. Suppose that the $(\mc E,\mc F)\in{\bf Comp}(\hil,\mc K_1,\mc K_2)$ has the joint channel $\mc M\in{\bf Ch}(\hil,\mc K_1\otimes\mc K_2)$ and pick $\mc G\in{\bf Ch}(\hil',\hil)$. It follows that also the pair $(\mc E\circ\mc G,\mc F\circ\mc G)$ is compatible since it has (among others) the joint channel $\mc M\circ\mc G$. This means that, whenever the pair $(\mc E,\mc F)$ is compatible and $(\mc E',\mc F')\leq_\mr{prae}(\mc E,\mc F)$, then also $(\mc E',\mc F')$ is compatible.

When $\mc E'\in{\bf Ch}(\hil,\mc K_1')$, $\mc F'\in{\bf Ch}(\hil,\mc K_2')$, $\mc E\in{\bf Ch}(\hil,\mc K_1)$, and $\mc F\in{\bf Ch}(\hil,\mc K_2)$, we denote $(\mc E',\mc F')\leq_\mr{post}(\mc E,\mc F)$ if there are channels $\mc A\in{\bf Ch}(\mc K_1,\mc K_1')$ and $\mc B\in{\bf Ch}(\mc K_2,\mc K_2')$ such that $\mc E'=\mc A\circ\mc E$ and $\mc F'=\mc B\circ\mc F$. We denote the equivalence relation corresponding to the partial order $\leq_\mr{post}$ by $=_\mr{post}$. If $(\mc E,\mc F)\in{\bf Comp}(\hil,\mc K_1,\mc K_2)$ has the joint channel $\mc M\in{\bf Ch}(\hil,\mc K_1\otimes\mc K_2)$ and we choose $\mc A\in{\bf Ch}(\mc K_1,\mc K_1')$ and $\mc B\in{\bf Ch}(\mc K_2,\mc K_2')$ we may define the joint channel $(\mc A\otimes\mc B)\circ\mc M$ for the pair $(\mc A\circ\mc E,\mc B\circ\mc F)$. Thus, whenever the pair $(\mc E,\mc F)$ is compatible and $(\mc E',\mc F')\leq_\mr{post}(\mc E,\mc F)$, then also $(\mc E',\mc F')$ is compatible.

As for observables, we may easily prove the following (the robustness measure $w$ possesses the same properties):
\begin{theorem}\label{theor:prop2}
Let $(\mc E,\mc F)$, $(\mc C,\mc D)$, and $(\mc C',\mc D')$ be channel pairs.
\begin{itemize}
\item[(a)] If $(\mc C,\mc D)\leq_\mr{prae}(\mc E,\mc F)$, then $W(\mc C,\mc D)\geq W(\mc E,\mc F)$, and if $(\mc C,\mc D)=_\mr{prae}(\mc E,\mc F)$, then $W(\mc C,\mc D)=W(\mc E,\mc F)$.
\item[(b)] If $(\mc C',\mc D')\leq_\mr{post}(\mc E,\mc F)$, then $W(\mc C',\mc D')\geq W(\mc E,\mc F)$, and if $(\mc C',\mc D')=_\mr{post}(\mc E,\mc F)$, then $W(\mc C',\mc D')=W(\mc E,\mc F)$.
\end{itemize}
\end{theorem}
Thus, especially, we have $W(\mc E,\mc F)\geq W(\mr{id})$ for any $\mc E\in{\bf Ch}(\hil,\mc K_1)$ and $\mc F\in{\bf Ch}(\hil,\mc K_2)$, where $\mr{id}\in{\bf Ch}(\hil,\hil)$ is the identity channel, i.e.,\ the pair $(\mr{id},\mr{id})$ is the most incompatible pair of channels with respect to the robustness measures. The robustness measure $W$ (as well as $w$) attains the same minimal value at any channel pair in the post-processing equivalence class determined by the identity channel pair $(\mr{id},\mr{id})$, $\mr{id}\in{\bf Ch}(\hil,\hil)$. Clearly, a pair $(\mc E,\mc F)\in{\bf Ch}(\hil,\mc K_1)\times{\bf Ch}(\hil,\mc K_2)$ is in this equivalence class when they are left-invertible by channels, i.e.,\ there are channels $\mc A\in{\bf Ch}(\mc K_1,\hil)$ and $\mc B\in{\bf Ch}(\mc K_2,\hil)$ such that $\mc A\circ\mc E=\mc B\circ\mc F=\mr{id}$. From now on, we call such channels {\it decodable}. As a special case of \cite[Corollary 1]{JencovaPetz2006}, when $\hil$ and $\mc K$ are finite dimensional, a channel $\mc E\in{\bf Ch}(\hil,\mc K)$ is decodable if and only if there is a Hilbert space $\mc K_0$, a unitary operator $U:\hil\otimes\mc K_0\to\mc K$, and a positive trace-1 operator $T$ on $\mc K_0$ such that $\mc E(\rho)=U(\rho\otimes T)U^*$ for all $\rho\in\mc S(\hil)$. It follows that a channel with unitarily equivalent input and output spaces is decodable if and only if it is a unitary channel, i.e.,\ of the form $\rho\mapsto U\rho U^*$ with a unitary operator $U$. The decodable channels posses essentially the same properties with respect to the robustness measures as the identity channel.

Let us fix Hilbert spaces $\hil,\,\hil',\,\mc K$, and $\mc K'$ and standard Borel spaces $(\Om,\Sigma)$ and $(\Om',\Sigma')$. We denote, e.g.,\ for $\hil,\,\mc K$, and $\Sigma$, by ${\bf Comp}(\Sigma,\hil,\mc K)$ the set of compatible pairs in ${\bf Obs}(\Sigma,\hil)\times{\bf Ch}(\hil,\mc K)$.

We denote $(\ms M',\mc E')\leq_\mr{prae}(\ms M,\mc E)$ for $\ms M'\in{\bf Obs}(\Sigma,\hil')$, $\mc E'\in{\bf Ch}(\hil',\mc K)$, $\ms M\in{\bf Obs}(\Sigma,\hil)$, and $\mc E\in{\bf Ch}(\hil,\mc K)$ if there is $\mc G\in{\bf Ch}(\hil',\hil)$ such that $\ms M'=\mc G^*\circ\ms M$ and $\mc E'=\mc E\circ\mc G$. Again, the equivalence relation corresponding to the partial order $\leq_\mr{prae}$ is $=_\mr{prae}$. It is easy to see that, if $(\ms M,\mc E)$ is a compatible pair and $(\ms M',\mc E')\leq_\mr{prae}(\ms M,\mc E)$, then $(\ms M',\mc E')$ is compatible as well.

If $\ms M'=\ms M^\beta$ and $\mc E'=\mc B\circ\mc E$, where $\ms M'\in{\bf Obs}(\Sigma',\hil)$, $\mc E'\in{\bf Ch}(\hil,\mc K')$, $\ms M\in{\bf Obs}(\Sigma,\hil)$, and $\mc E\in{\bf Ch}(\hil,\mc K)$, for a channel $\mc B\in{\bf Ch}(\mc K,\mc K')$ and a Markov kernel $\beta:\Sigma'\times\Om\to\mb R$, we denote $(\ms M',\mc E')\leq_\mr{post}(\ms M,\mc E)$. The equivalence relation associated with the partial order $\leq_\mr{post}$ is denoted by $=_\mr{post}$. Suppose that $(\ms M,\mc E)\in{\bf Comp}(\Sigma,\hil,\mc K)$ has the joint instrument $\Gamma\in{\bf Ins}(\Sigma,\hil,\mc K)$ and pick $\mc B\in{\bf Ch}(\mc K,\mc K')$ and a Markov kernel $\beta:\Sigma'\times\Om\to\mb R$. It is straight-forward to check that $\Gamma'\in{\bf Ins}(\Sigma',\hil,\mc K')$,
$$
\Gamma'_Y(\rho)=\mc B\bigg(\int_\Om\beta(Y,\om)\Gamma_{d\om}(\rho)\bigg),\qquad Y\in\Sigma',\quad\rho\in\mc S(\hil),
$$
is a joint instrument for $(\ms M^\beta,\mc B\circ\mc E)$ implying that, whenever the pair $(\ms M,\mc E)$ is compatible and $(\ms M',\mc E')\leq_\mr{post}(\ms M,\mc E)$, then $(\ms M',\mc E')$ is compatible as well.

Again, one easily proves the following properties (which also hold for $w$):
\begin{theorem}\label{theor:prop3}
Let $(\ms M,\mc E)$, $(\ms N,\mc F)$, and $(\ms N',\mc F')$ be observable-channel pairs.
\begin{itemize}
\item[(a)] If $(\ms N,\mc F)\leq_\mr{prae}(\ms M,\mc E)$, then $W(\ms N,\mc F)\geq W(\ms M,\mc E)$, and if $(\ms N,\mc F)=_\mr{prae}(\ms M,\mc E)$, then $W(\ms N,\mc F)=W(\ms M,\mc E)$.
\item[(b)] If $(\ms N',\mc F')\leq_\mr{post}(\ms M,\mc E)$, then $W(\ms N',\mc F')\geq W(\ms M,\mc E)$, and if $(\ms N',\mc F')=_\mr{post}(\ms M,\mc E)$, then $W(\ms N',\mc F')=W(\ms M,\mc E)$.
\end{itemize}
\end{theorem}

Theorems \ref{theor:properties}, \ref{theor:prop2}, and \ref{theor:prop3} together tell that instead of considering robustness measures as functions on individual device pairs, they can be defined on pre- or post-processing equivalence classes. The partial orders invoked by pre- and post-processing in the set of observables and their meaning are studied, e.g.,\ in \cite{esijarj, heinonen05}. The results of this section also tell that the measure $R:=1/W-1$ (as well as $1/w-1$) is an incompatibility monotone from the perspective of \cite{HeKiRe15}. The operations defining the preorders $\leq_\mr{prae}$ and $\leq_\mr{post}$, common pre-processing and independent bipartite post-processing, can be naturally viewed as compatibility non-decreasing maps and any incompatibility measure should naturally behave monotonously under these operations. Monotonicity under $\leq_\mr{prae}$ has been required already in \cite{HeKiRe15, Pusey15}.

\section{Examples}\label{sec:ex}

In the remainder of this article, we calculate the robustness of incompatibility $W$ for three special cases: the finite dimensional Weyl pair, the pair of decodable, hence especially unitary, channels, and the pair consisting of a rank-1 sharp observable (von Neumann observable) and a decodable channel. In each case, the quantity $R=1/W-1$ measures how well the pair resists joining under noise. Hence, in the first case, we essentially determine the overall resistance to joint measuring of a finite-dimensional `position-momentum' pair. In the second case, we find how well (or how poorly) we may approximately combine a pair of decodable channels in a single channel. The third case enlightens the issue of how close can the total state change associated with an approximate measurement of a von Neumann observable be to an information-preserving channel

\subsection{Robustness of incompatibility for a sharp Weyl pair}\label{sec:weyl}

In this section, we calculate the robustness of incompatibility for a particular pair of incompatible observables: a finite-dimensional Weyl pair. Let us fix a $d$-dimensional Hilbert space $\hil$ ($d<\infty$) which has the orthonormal base $\{\f_j\}_{j\in\mb Z_d}$. We denote $\dual{j}{k}=e^{i2\pi jk/d}$ for all $j,\,k\in\mb Z_d$ and define the linear operator $\mc F\in\mc L(\hil)$ through
\begin{equation}\label{eq:fourier}
\mc F\f_j=\frac{1}{\sqrt{d}}\sum_{i\in\mb Z_d}\ovl{\dual{i}{j}}\f_i,\qquad j\in\mb Z_d.
\end{equation}
This operator is the Fourier-operator and its adjoint is defined through
$$
\mc F^*\f_j=\frac{1}{\sqrt{d}}\sum_{i\in\mb Z_d}\dual{i}{j}\f_i,\qquad j\in\mb Z_d.
$$

For simplicity, we denote by ${\bf Obs}_d$ the set of observables operating in $\hil$ whose value space is $\mb Z_d$ (equipped with its power set as the $\sigma$-algebra). Hence, an observable $\ms M\in{\bf Obs}_d$ is defined by the values $\ms M(\{j\}):=\ms M_j\in\mc L(\hil)$, $j\in\mb Z_d$, and we write $\ms M=(\ms M_j)_{j\in\mb Z_d}$. We denote the set of compatible pairs in ${\bf Obs}_d\times{\bf Obs}_d$ by ${\bf JM}_d$. We denote the set of $\mb Z_d\times\mb Z_d$-valued observables operating in $\hil$ (the possible joint observables for the compatible pairs $(\ms M,\ms N)\in{\bf JM}_d$) by ${\bf Obs}_{d\times d}$. When $\ms G\in{\bf Obs}_{d\times d}$, we set $\ms G_{j,k}:=\ms G\big(\{(j,k)\}\big)$ for all $j,\,k\in\mb Z_d$.

We fix another orthonormal basis $\{\psi_k\}_{k\in\mb Z_d}$ by setting $\psi_k=\mc F^*\f_k$. It follows that $\sis{\f_j}{\psi_k}=d^{-1/2}$ for all $j,\,k\in\mb Z_d$, so that the bases $\{\f_j\}$ and $\{\psi_k\}$ are an example of a pair of mutually unbiased bases. Let us denote
$$
\ms Q_j:=|\f_j\ra\la\f_j|,\quad\ms P_k:=|\psi_k\ra\la\psi_k|,\qquad j,\,k\in\mb Z_d
$$
and define the sharp observables $\ms Q:=(\ms Q_j)_{j\in\mb Z_d}\in{\bf Obs}_d$ and $\ms P:=(\ms P_k)_{k\in\mb Z_d}\in{\bf Obs}_d$.

For each $q,\,p\in\mb Z_d$, we may define the operators $U_q,\,V_p,\,W_{q,p}\in\mc L(\hil)$ through
\begin{eqnarray}
U_q\f_j&=&\f_{j+q}, \label{eq:Usiirto}\\
V_p\f_j&=&\dual{j}{p}\f_j, \label{eq:Vsiirto}\\
W_{q,p}&=&U_qV_p, \label{eq:Weyl}
\end{eqnarray}
where the sums and differences are considered as cyclic on $\mb Z_d$. Thus $(q,p)\mapsto W_{q,p}$ is a projective unitary representation of $\mb Z_d\times\mb Z_d$ in $\hil$ which we call as the {\it $d$-dimensional Weyl representation}. It follows that
$$
W_{q,p}^*\ms Q_jW_{q,p}=\ms Q_{j-q},\quad W_{q,p}^*\ms P_kW_{q,p}=\ms P_{k-p},\qquad j,\,k,\,q,\,p\in\mb Z_d,
$$
i.e.,\ $\ms Q$ and $\ms P$ are {\it Weyl-covariant}.

We denote the set of all Weyl-covariant pairs $(\ms A,\ms B)\in{\bf Obs}_d\times{\bf Obs}_d$, i.e.,
\begin{equation}\label{eq:Wkov}
W_{q,p}^*\ms A_jW_{q,p}=\ms A_{j-q},\quad W_{q,p}^*\ms B_kW_{q,p}=\ms B_{k-p},\qquad j,\,k,\,q,\,p\in\mb Z_d,
\end{equation}
by ${\bf Obs}_{d\times d}^W$. Any Weyl-covariant pair $(\ms A,\ms B)$, where $\ms A$ and $\ms B$ are sharp, is unitarily equivalent with the fixed pair $(\ms Q,\ms P)$ in the sense that there is a unitary operator $U$ on $\hil$ such that $\ms A_j=U^*\ms Q_jU$ and $\ms B_k=U^*\ms P_kU$ for all $j,\,k\in\mb Z_d$. From now on, we call the pair $(\ms Q,\ms P)$ as the {\it Weyl pair}.

If $(\ms M,\ms N)\in{\bf Obs}_{d\times d}^W$, there are probability distributions $\mu=(\mu_j)_{j\in\mb Z_d}$ and $\nu=(\nu_k)_{k\in\mb Z_d}$ such that $\ms M=\mu*\ms Q$ and $\ms N=\nu*\ms P$, i.e.,
\begin{equation}\label{eq:convolution}
\ms M_j=\sum_{q\in\mb Z_d}\mu_{j-q}\ms Q_q,\quad\ms N_k=\sum_{p\in\mb Z_d}\nu_{k-p}\ms P_p
\end{equation}
for all $j,\,k\in\mb Z_d$. Moreover, such a Weyl-covariant pair is jointly measurable if and only if there is a state $\rho\in\mc S(\hil)$ such that
\begin{equation}\label{eq:tilaehto}
\mu_j=\tr{\rho\ms Q_{-j}},\quad\nu_k=\tr{\rho\ms P_{-k}},\qquad j,\,k\in\mb Z_d.
\end{equation}
The latter condition can also be written using a purification $\eta\in\hil\otimes\hil$ of $\rho$, so that
\begin{equation}\label{eq:vektoriehto}
\mu_j=\sis{\eta}{(\ms Q_{-j}\otimes\id)\eta},\quad\nu_k=\sis{\eta}{(\ms P_{-k}\otimes\id)\eta},\qquad j,\,k\in\mb Z_d.
\end{equation}
For proofs of these facts about Weyl-covariant pairs, we refer to \cite{Heinosaari_etal2012}.

The following lemma is useful for evaluating the robustness of incompatibility for any Weyl-covariant pair.

\begin{lemma}\label{lemma:covopt}
Let $(\ms A,\ms B)\in{\bf Obs}_{d\times d}^W$. One has
$$
W(\ms A,\ms B)=\sup_{(\ms M,\ms N)\in{\bf Obs}_{d\times d}^W}w(\ms A,\ms B|\ms M,\ms N).
$$
\end{lemma}

\begin{proof}
Let us first define a map ${\bf Obs}_{d\times d}\ni\ms G\mapsto\ms G^W\in{\bf Obs}_{d\times d}$ by setting
$$
\ms G^W_{j,k}=\frac{1}{d^2}\sum_{q,\,p\in\mb Z_d}W_{q,p}^*\ms G_{j+q,k+p}W_{q,p},\quad j,\,k\in\mb Z_d.
$$
It follows (as one may easily check) that
$$
W_{q,p}^*\ms G^W_{j,k}W_{q,p}=\ms G^W_{j-q,k-p}
$$
for all $j,\,k,\,q,\,p\in\mb Z_d$. Similarly for any $\ms M\in{\bf Obs}_d$, we define $\ms M^{W,1},\,\ms M^{W,2}\in{\bf Obs}_d$ through
\begin{eqnarray*}
\ms M^{W,1}_j&=&\frac{1}{d^2}\sum_{q,p\in\mb Z_d}W_{q,p}^*\ms M_{j+q}W_{q,p},\quad j\in\mb Z_d\\
\ms M^{W,2}_k&=&\frac{1}{d^2}\sum_{q,p\in\mb Z_d}W_{q,p}^*\ms M_{k+p}W_{q,p},\quad k\in\mb Z_d
\end{eqnarray*}
It follows that, when $\ms G\in{\bf Obs}_{d\times d}$ is a joint observable for $(\ms M,\ms N)\in{\bf JM}_d$, i.e.,\ $\sum_k\ms G_{j,k}=\ms M_j$ and $\sum_j\ms G_{j,k}=\ms N_k$, then $(\ms M^{W,1},\ms N^{W,2})\in{\bf JM}_d$ and this pair has (among others) the joint observable $\ms G^W$. Furthermore, $(\ms M^{W,1},\ms N^{W,2})\in{\bf Obs}_{d\times d}^W$ for any $\ms M,\,\ms N\in{\bf Obs}_d$, and, if $(\ms A,\ms B)\in{\bf Obs}_{d\times d}^W$, then $\ms A^{W,1}=\ms A$ and $\ms B^{W,2}=\ms B$.

Suppose now that $t<W(\ms A,\ms B)$ and let $(\ms M,\ms N)\in{\bf Obs}_d\times{\bf Obs}_d$, and $(\tilde{\ms M},\tilde{\ms N})\in{\bf JM}_d$ be such that
$$
t(\ms A,\ms B)+(1-t)(\ms M,\ms N)=(\tilde{\ms M},\tilde{\ms N}).
$$
Since $(\ms A,\ms B)$ is Weyl covariant, it follows that
$$
t(\ms A,\ms B)+(1-t)(\ms M^{W,1},\ms N^{W,2})=(\tilde{\ms M}^{W,1},\tilde{\ms N}^{W,2}),
$$
where $(\tilde{\ms M}^{W,1},\tilde{\ms N}^{W,2})\in{\bf JM}_d$, since,\ if $\ms G$ is a joint observable for $(\tilde{\ms M},\tilde{\ms N})$, then $\ms G^W$ is a joint observable for $(\tilde{\ms M}^{W,1},\tilde{\ms N}^{W,2})$. Hence, for all $t<W(\ms A,\ms B)$, we find $(\ms M,\ms N)\in{\bf Obs}_{d\times d}^W$ and $(\tilde{\ms M},\tilde{\ms N})\in{\bf JM}_d\cap{\bf Obs}_{d\times d}^W$ such that
$$
t(\ms A,\ms B)+(1-t)(\ms M,\ms N)=(\tilde{\ms M},\tilde{\ms N}),
$$
and the claim is proven.
\end{proof}

\begin{theorem}\label{theor:Wrob}
The robustness of incompatibility for the sharp Weyl pair is given by
\begin{equation}\label{eq:Wrob}
W(\ms Q,\ms P)=\frac12\Big(1+\frac{1}{\sqrt{d}}\Big).
\end{equation}
\end{theorem}

\begin{proof}
Let $t<W(\ms Q,\ms P)$ and, using Lemma \ref{lemma:covopt}, suppose that $(\ms M,\ms N)\in{\bf Obs}_{d\times d}^W$ and $(\ms A,\ms B)\in{\bf JM}_d$ are such that
$$
(\ms A,\ms B)=t(\ms Q,\ms P)+(1-t)(\ms M,\ms N).
$$
Clearly, $(\ms A,\ms B)\in{\bf Obs}_{d\times d}^W$. Let $\mu$ and $\nu$ be probability distributions such that $\ms M$ and $\ms N$ are given by (\ref{eq:convolution}). Define $\delta$ to be the probability distribution having $\delta_0=1$ (and, of course, $\delta_j=0$ for $j\neq 0$). One may write
$$
\ms A=(t\delta+(1-t)\mu)*\ms Q,\quad \ms B=(t\delta+(1-t)\nu)*\ms P.
$$
Hence, there has to be $\eta\in\hil\otimes\hil$ such that
\begin{eqnarray}
\sis{\eta}{(\ms Q_{-j}\otimes\id)\eta}&=&t\delta_j+(1-t)\mu_j,\quad j\in\mb Z_d \label{eq:abehto1}\\
\sis{\eta}{(\ms P_{-k}\otimes\id)\eta}&=&t\delta_k+(1-t)\nu_k,\quad k\in\mb Z_d. \label{eq:abehto2}
\end{eqnarray}

We may write $\eta=\sum_j\f_j\otimes\zeta_j=\sum_k\psi_k\otimes\xi_k$ for some $\zeta_j,\,\xi_k\in\hil$, $j,\,k\in\mb Z_d$. Following the procedure carried out in the proof of \cite[Lemma 1]{Heinosaari_etal2012}, one obtains the (tight) inequalities
\begin{eqnarray}
\sqrt{t(1-\nu_0)+\nu_0}&\leq&\frac{1}{\sqrt{d}}\big(\sqrt{t(1-\mu_0)+\mu_0}+\alpha\sqrt{1-t}\big) \label{eq:epayht1}\\
\sqrt{t(1-\mu_0)+\mu_0}&\leq&\frac{1}{\sqrt{d}}\big(\sqrt{t(1-\nu_0)+\nu_0}+\beta\sqrt{1-t}\big), \label{eq:epayht2}
\end{eqnarray}
where $\alpha=\sum_{j\neq0}\sqrt{\mu_j}$ and $\beta=\sum_{k\neq0}\sqrt{\nu_k}$. It is easy to see that $\alpha\leq\sqrt{d-1}\sqrt{1-\mu_0}$ and $\beta\leq\sqrt{d-1}\sqrt{1-\nu_0}$, and these bounds are reached when $\mu_j=(1-\mu_0)/(d-1)$ and $\nu_k=(1-\nu_0)/(d-1)$ for all $j,\,k\neq0$. Solving from (\ref{eq:epayht1})-(\ref{eq:epayht2}), one obtains the following inequalities:
$$
t\leq\frac{(\alpha+\sqrt{d}\beta)^2-(d-1)^2\mu_0}{(\alpha+\sqrt{d}\beta)^2+(d-1)^2(1-\mu_0)}\leq\frac{(d-1)\big(\sqrt{1-\mu_0}+\sqrt{d(1-\nu_0)}\big)^2-(d-1)^2\mu_0}{(d-1)\big(\sqrt{1-\mu_0}+\sqrt{d(1-\nu_0)}\big)^2+(d-1)^2(1-\mu_0)}.
$$
It is easy to see that as one lets $\mu_0,\,\nu_0\downarrow0$, the latter bound increases and setting $\mu_0=\nu_0=0$, one obtains $W(\ms Q,\ms P)\leq\frac12(1+1/\sqrt{d})$.

It remains to be shown that the bound obtained above is reachable. Setting $t=\frac12(1+1/\sqrt{d})$,
$$
\eta=\sqrt{\frac{\sqrt{d}}{2(\sqrt{d}+1)}}(\f_0+\psi_0)\otimes\xi
$$
for any unit vector $\xi\in\hil$, $\mu_0=\nu_0=0$, and $\mu_j=\nu_j=1/(d-1)$, for $j\neq0$, one finds that Equations (\ref{eq:abehto1})-(\ref{eq:abehto2}) hold (and in (\ref{eq:epayht1})-(\ref{eq:epayht2}) the inequalities can both be replaced by equalities). Hence, the claim is proven.
\end{proof}

\subsection{Robustness of incompatibility for a pair of decodable channels}

Let us fix a finite-dimensional Hilbert space $\hil$, $\dim{\hil}=d$. We denote ${\bf Ch}(\hil,\hil)={\bf Ch}_d$ and ${\bf Ch}(\hil,\hil\otimes\hil)={\bf Ch}_{d\times d}$. The set of compatible pairs within ${\bf Ch}_d\times{\bf Ch}_d$ is denoted by ${\bf Comp}_d$. Moreover, $\mr{id}$ stands for the identity channel $\mc S(\hil)\to\mc S(\hil)$, i.e.,\ $\mr{id}(\rho)=\rho$ for all $\rho\in\mc S(\hil)$. Clearly, the dual $\mr{id}^*$ is the identity map on $\mc L(\hil)$ which we denote by $\mr{id}$ as well.

We fix an orthonormal basis $\{|n\ra\}_{n=1}^d$ for the duration of this subsection and introduce the rank-1 operators $\Om_d=\sum_{m,n=1}^d|mm\ra\la nn|\in\mc L(\hil\otimes\hil)$ and $\Om_{d\otimes d}=\sum_{j,k,m,n=1}^d|jkjk\ra\la mnmn|\in\mc L(\hil\otimes\hil\otimes\hil\otimes\hil)$, where $|m_1\cdots m_n\ra:=|m_1\ra\otimes\cdots\otimes|m_n\ra$. For any $\mc E\in{\bf Ch}_d$ (respectively $\mc F\in{\bf Ch}_{d\times d}$) we define the {\it Choi operator} $M(\mc E)=(\mc E^*\otimes\mr{id})(\Om_d)\in\mc L(\hil\otimes\hil)$ (respectively $M(\mc F)=(\mc F^*\otimes\mr{id}\otimes\mr{id})(\Om_{d\times d})\in\mc L(\hil\otimes\hil\otimes\hil)$). We denote the transpose of $B\in\mc L(\hil)$ with respect to the fixed basis by $B^T$ and denote $\ovl B:=B^{T*}$. Furthermore, we denote the partial transpose restricted to the subsystems 2 and 3 of $C\in\mc L(\hil\otimes\hil\otimes\hil)$ by $C^\Gamma$, i.e.,\ when $C_1,\,C_2,\,C_3\in\mc L(\hil)$, we have $(C_1\otimes C_2\otimes C_3)^\Gamma=C_1\otimes C_2^T\otimes C_3^T$. 

If $\mc E\in{\bf Ch}_d$ (respectively $\mc F\in{\bf Ch}_{d\times d}$) is such that $\mc E(U\rho U^*)=U\mc E(\rho)U^*$ (respectively $\mc F(U\rho U^*)=(U\otimes U)\mc F(\rho)(U\otimes U)^*$) for all $\rho\in\mc S(\hil)$ and all unitary $U\in\mc L(\hil)$, we say that $\mc E$ (respectively $\mc F$) is {\it fully covariant} and denote $\mc E\in{\bf Cov}_d$ (respectively $\mc F\in{\bf Cov}_{d\times d}$). Denote by $dU$ the normalized Haar measure of the (compact) unitary group $U(d)$. We may define the map ${\bf Ch}_d\ni\mc E\mapsto\mc E_{av}\in{\bf Cov}_d$ by setting
$$
\mc E_{av}(\rho)=\int_{U(d)}U^*\mc E(U\rho U^*)U\,dU,\qquad \rho\in\mc S(\hil).
$$
Likewise, one can set up a map ${\bf Ch}_{d\times d}\ni\mc F\mapsto\mc F_{av}\in{\bf Cov}_{d\times d}$ through
\begin{equation}\label{eq:Fav}
\mc F_{av}(\rho)=\int_{U(d)}(U\otimes U)^*\mc F(U\rho U^*)(U\otimes U)\,dU,\qquad\rho\in\mc S(\hil).
\end{equation}
The sets ${\bf Cov}_d$ and ${\bf Cov}_{d\times d}$ coincide with the sets of the fixed points of these maps.

Suppose that $M\in\mc L(\hil\otimes\hil\otimes\hil)$ is a positive operator whose partial trace over the subsystems 2 and 3 coincides with $\id_\hil$, i.e.,\ $M$ is a Choi operator of a channel $\mc F\in{\bf Ch}_{d\times d}$. We may define the operator $M_{av}\in\mc L(\hil\otimes\hil\otimes\hil)$ through
\begin{equation}\label{eq:Mav}
M_{av}=\int_{U(d)}(U\otimes\ovl U\otimes\ovl U)M(U\otimes\ovl U\otimes\ovl U)^*\,dU.
\end{equation}
We have that $(U\otimes\ovl U\otimes\ovl U)M=M(U\otimes\ovl U\otimes\ovl U)$ for all $U\in U(d)$ if and only if $M=M_{av}$. It is straightforward to check that
$$
M(\mc F_{av})=M(\mc F)_{av},\qquad \mc F\in{\bf Ch}_{d\times d}.
$$

\begin{lemma}\label{lemma:cov}
Let $\mc E\in{\bf Cov}_d$. The robustness of self-incompatibility for $\mc E$ is given by
$$
W(\mc E)=\sup_{\mc C\in{\bf Cov}_d}w(\mc E,\mc E|\mc C,\mc C).
$$
\end{lemma}

\begin{proof}
Let $t<W(\mc E)$ and suppose that $(\mc A,\mc B)\in{\bf Ch}_d\times{\bf Ch}_d$ and $(\mc A',\mc B')\in{\bf Comp}_d$ are such that
$$
t(\mc E,\mc E)+(1-t)(\mc A,\mc B)=(\mc A',\mc B')
$$
It follows that $(\mc A'_{av},\mc B'_{av})\in{\bf Comp}_d$ as well, since if $\mc F$ is a joint channel for $(\mc A',\mc B')$, it is immediate that $\mc F_{av}$ is a joint channel for $(\mc A'_{av},\mc B'_{av})$. Hence,
$$
t(\mc E,\mc E)+(1-t)(\mc A_{av},\mc B_{av})=(\mc A'_{av},\mc B'_{av}).
$$
Denote $\mc C=\frac12(\mc A_{av}+\mc B_{av})$ and $\mc C'=\frac12(\mc A'_{av}+\mc B'_{av})$. Again it easily follows
$$
t(\mc E,\mc E)+(1-t)(\mc C,\mc C)=(\mc C',\mc C').
$$
Denote by $F\in\mc L(\hil\otimes\hil)$ the flip operator, $F(\f\otimes\psi)=\psi\otimes\f$ for all $\f,\,\psi\in\hil$. The pair $(\mc C',\mc C')$ is compatible, since if $M$ is the Choi operator of a joint channel of $(\mc A'_{av},\mc B'_{av})$, then $\frac12\big(M+(\id\otimes F)M(\id\otimes F)\big)$ is the Choi operator for a joint channel for $(\mc C',\mc C')$.
\end{proof}

Since $\mr{id}$ is fully covariant, the preceding lemma restricts the problem of evaluating the robustness of self-incompatibility of $\mr{id}$ (quite considerably, as we will see). 
According to Lemma \ref{lemma:cov} (and its proof), $W(\mr{id})$ is simply the supremum of those $t\in[0,1]$ such that $t\,\mr{id}+(1-t)\mc E$ is self-compatible for some $\mc E\in{\bf Cov}_d$. Next, we determine the set of self-compatible fully covariant channels which essentially resolves the problem of determining $W(\mr{id})$.

Let us fix a self-compatible $\mc A\in{\bf Cov}_d$ and a joint channel $\mc F\in{\bf Ch}_{d\times d}$ for the pair $(\mc A,\mc A)$. Since $\mc A$ is fully covariant, the channel $\mc F_{av}$ still has the same marginals. Let $M\in\mc L(\hil\otimes\hil\otimes\hil)$ be the Choi operator of $\mc F$ (with respect to our fixed basis). Hence, $M_{av}$ is the Choi operator of $\mc F_{av}$. From (\ref{eq:Mav}) it follows that
$$
M_{av}=\bigg(\int(U\otimes U\otimes U)M^\Gamma(U\otimes U\otimes U)^*\,dU\bigg)^\Gamma.
$$
This means that $(U\otimes U\otimes U)M_{av}^\Gamma=M_{av}^\Gamma(U\otimes U\otimes U)$ for all $U\in U(d)$, i.e.,\ $M_{av}^\Gamma$ is $U\otimes U\otimes U$-invariant.

For any permutation $\pi$ of three elements, denote by $V_\pi\in\mc L(\hil\otimes\hil\otimes\hil)$ the unitary operator defined through
$$
V_\pi(\f_1\otimes\f_2\otimes\f_3)=\f_{\pi^{-1}(1)}\otimes\f_{\pi^{-1}(2)}\otimes\f_{\pi^{-1}(3)}
$$
for all $\f_1,\,\f_2,\,\f_3\in\hil$. A well-known result from Weyl states that any $U\otimes U\otimes U$-invariant operator is a linear combination of the permutation operators $V_\pi$ from which it follows immediately that our averaged Choi operator can be expressed as a linear combination
$$
M_{av}=\sum_\pi\lambda_\pi V_\pi^\Gamma,
$$
where $\lambda_\pi$ are complex numbers. However, we must make sure that this linear combination is a positive operator whose partial trace over the subsystems 2 and 3 is $\id_\hil$. To this end, we must express the linear combination in a more revealing form.

As the commutant of the set $U\otimes\ovl U\otimes\ovl U$, $U\in U(d)$, the operator system spanned by the six operators $V_\pi^\Gamma$ is a 6-dimensional algebra with the exception in case $d=2$, when the algebra is 5-dimensional. In \cite{EggelingWerner2001}, this algebra was shown to have the basis consisting of the operators
\begin{eqnarray*}
S_\pm&=&\frac12\big(\id\pm V_{(23)}^\Gamma-\frac{1}{d\pm1}(V_{(12)}^\Gamma+V_{(13)}^\Gamma\pm V_{(123)}^\Gamma\pm V_{(132)}^\Gamma)\big),\\
S_0&=&\frac{1}{d^2-1}\big(d(V_{(12)}^\Gamma+V_{(13)}^\Gamma)-(V_{(123)}^\Gamma+V_{(132)}^\Gamma)\big),\\
S_1&=&\frac{1}{d^2-1}\big(d(V_{(123)}^\Gamma+V_{(132)}^\Gamma)-(V_{(12)}^\Gamma+V_{(13)}^\Gamma)\big),\\
S_2&=&\frac{1}{\sqrt{d^2-1}}(V_{(12)}^\Gamma-V_{(13)}^\Gamma),\\
S_3&=&\frac{i}{\sqrt{d^2-1}}(V_{(123)}^\Gamma-V_{(132)}^\Gamma),
\end{eqnarray*}
where $S_\pm,\,S_0$ are mutually orthogonal projections summing up to $\id$ and $S_1,\,S_2$, and $S_3$ are selfadjoint operators supported on the eigenspace of $S_0$ that are interrelated in the same way as the Pauli matrices. Hence, formally
$S_jS_\pm=S_\pm S_j=0$ and $S_j^2=S_0$ for all $j=0,\,1,\,2,\,3$, and $S_1S_2=iS_3$ with cyclic permutations. Note that, in the case $d=2$, $S_-=0$. It follows that a linear combination $\mu_+S_++\mu_-S_-+\mu_0S_0+\mu_1S_1+\mu_2S_2+\mu_3S_3$ is positive if and only if the multipliers are real, $\mu_\pm,\,\mu_0\geq0$, and $\mu_1^2+\mu_2^2+\mu_3^2\leq\mu_0^2$.

Let us now impose our additional symmetry condition, i.e.,\ the marginals of $\mc F_{av}$ must coincide or, equivalently, $V_{(23)}M_{av}=M_{av}V_{(23)}$; note that $V_{(23)}^\Gamma=V_{(23)}$. One easily sees that this requirement necessitates that the multipliers of $S_2$ and $S_3$ in $M$ be zero. Moreover, through direct calculation, one finds that for the partial traces over the subsystems 2 and 3,
$$
\begin{array}{rclcrcl}
\mr{tr}_{23}[S_+]&=&\frac12(d-1)(d+2)\id_\hil,&\qquad&\mr{tr}_{23}[S_-]&=&\frac12(d+1)(d-2)\id_\hil,\\
\mr{tr}_{23}[S_0]&=&2\id_\hil,&\qquad&\mr{tr}_{23}[S_1]&=&0.
\end{array}
$$
Putting all this together, one finds that $M_{av}$ is of the form $M(t_+,t_-,t_0,t_1)$,
$$
M(t_+,t_-,t_0,t_1)=\frac{2}{(d-1)(d+2)}t_+S_++\frac{2}{(d+1)(d-2)}t_-S_-+\frac12(t_0S_0+t_1S_1),
$$
where $t_\pm,\,t_0\geq0$, $t_++t_-+t_0=1$ and $t_1\in[-t_0,t_0]$.

The set of the Choi operators $M(t_+,t_-,t_0,t_1)$ is a tetrahedron with the extreme points
\begin{eqnarray*}
M_\pm&=&\frac{2}{(d\mp1)(d\pm2)}S_\pm,\\
\tilde M_\pm&=&\frac12(S_0\pm S_1),
\end{eqnarray*}
and the partial traces over system 2 (or, equivalently, over system 3) of these are
\begin{eqnarray*}
\mr{tr}_2[M_\pm]&=&\frac{1}{d^2-1}(d\id_{\hil\otimes\hil}-\Om_d),\\
\mr{tr}_2[\tilde M_\pm]&=&\frac{1}{2(d\pm1)}\big(\id_{\hil\otimes\hil}+(d\pm2)\Om_d\big).
\end{eqnarray*}
Hence, the (coinciding) marginals of the channels corresponding to $M_\pm$ are $(d^2-1)^{-1}(d^2\mc T-\mr{id})$ and the (coinciding) marginals of the channels corresponding to $\tilde M_\pm$ are $\big(2(d\pm1)\big)^{-1}\big(d\mc T+(d\pm2)\mr{id}\big)$, where $\mc T:\mc S(\hil)\to\mc S(\hil)$ is the constant channel $\mc T(\rho)=d^{-1}\id_\hil$. Thus, the set of self-compatible fully covariant channels consists of the elements $\lambda\mc T+(1-\lambda)\mr{id}$ where $\frac{d}{2(d+1)}\leq\lambda\leq\frac{d^2}{d^2-1}$. It follows that $W(\mr{id})$ is the supremum of the $t\in[0,1]$ such that
$$
\frac{1}{2(1+d)}\big(d\mc T+(d+2)\mr{id}\big)-t\,\mr{id}
$$
is completely positive or, using the Choi operators, $\ptr{2}{\tilde M_+}-t\Om_d\geq0$. One easily finds that this condition is satisfied if and only if $t\leq\frac12(1+1/d)$. Thus, $W(\mr{id})=\frac12(1+1/d)$. According to the discussion following Theorem \ref{theor:prop2}, a pair $(\mc V,\mc W)\in{\bf Ch}(\hil,\mc K_1)\times{\bf Ch}(\hil,\mc K_2)$ of decodable channels has the same robustness of incompatibility as the pair of identity channels which is the minimum of the robustness measure, and hence:
\begin{theorem}\label{theor:selfinc}
Any pair $(\mc V,\mc W)\in{\bf Ch}(\hil,\mc K_1)\times{\bf Ch}(\hil,\mc K_2)$ of decodable channels minimizes $W$ amongst the channel pairs with the input space $\hil$, and this minimum value is
$$
W(\mc V,\mc W)=W(\mr{id})=\frac12\Big(1+\frac{1}{d}\Big),
$$
where $d=\dim{\hil}$ and $\mr{id}$ is the identity channel on $\mc S(\hil)$.
\end{theorem}

According to the discussion preceding Theorem \ref{theor:selfinc}, the self-compatible fully covariant channel $\mc A$ that has the optimality property
\begin{equation}\label{eq:idopt}
\mc A=\frac12\Big(1+\frac{1}{d}\Big)\,\mr{id}+\frac12\Big(1-\frac{1}{d}\Big)\mc E
\end{equation}
for some $\mc E\in{\bf Ch}_d$ is associated with the Choi operator $\ptr{2}{\tilde M_+}$ and is thus given by
$$
\mc A=\frac{d+2}{2(d+1)}\mr{id}+\frac{d}{2(d+1)}\mc T.
$$
One joint channel for the optimal pair $(\mc A,\mc A)$ is thus the optimal universal cloner $\rho\mapsto2(d+1)^{-1}S(\rho\otimes\id_\hil)S$, where $S\in\mc L(\hil\otimes\hil)$ is the orthogonal projector onto the symmetric subspace (the 1-eigenspace of the flip operator) of $\hil\otimes\hil$. In fact, the joint channel associated with the Choi operator $\tilde M_+$ is exactly this optimal cloner. The Choi operator associated with the channel $\mc E$ in the decomposition (\ref{eq:idopt}) is $\ptr{2}{M_+}=\ptr{2}{M_-}$, so that
$$
\mc E=-\frac{1}{d^2-1}\mr{id}+\frac{d^2}{d^2-1}\mc T.
$$
Especially $\mc E$ is self-compatible so that, in fact
$$
W(\mr{id})=w(\mr{id},\mr{id})=\frac12\Big(1+\frac{1}{d}\Big).
$$
The same holds for any pair $(\mc V,\mc W)$ of decodable channels. Moreover the pair
$$
\frac12\Big(1+\frac{1}{d}\Big)(\mc V,\mc W)+\frac12\Big(1-\frac{1}{d}\Big)(\mc E_{\mc V},\mc E_{\mc W})=\frac{d+2}{2(d+1)}(\mc V,\mc W)+\frac{d}{2(d+1)}(\mc T,\mc T)
$$
is compatible, where
$$
(\mc E_{\mc V},\mc E_{\mc W})=-\frac{1}{d^2-1}(\mc V,\mc W)+\frac{d^2}{d^2-1}(\mc T,\mc T).
$$

\subsection{Robustness of incompatibility for a von Neumann observable and a decodable channel}\label{sec:rank1id}

In this subsection, we calculate the robustness of incompatibility for a von Neumann observable, i.e.\ a rank-1 sharp observable (PVM), and a decodable channel. We start with the (slightly) simpler case where the decodable channel is simply the identity channel from which the more general case follows according to Section \ref{sec:robobs}.

For the remainder of this subsection, we fix a $d$-dimensional Hilbert space ($d<\infty$), and we fix an orthonormal basis $\{|n\ra\}_{n=0}^{d-1}\subset\hil$; we treat the index set $\{0,\,1,\ldots,\,d-1\}$ as the cyclic group $\mb Z_d$ and all sums and differences of these indices are considered cyclic. For simplicity, we denote the set of observables on the power set of $\mb Z_d$ and operating in $\hil$ by ${\bf Obs}_d$ and the set of channels ${\bf Ch}(\hil,\hil)$ by ${\bf Ch}_d$. We denote the set of instruments on the power set of $\mb Z_d$ and operating within $\mc S(\hil)$ by ${\bf Ins}_d$, i.e.,\ instruments in ${\bf Ins}_d$ are the possible joint instruments for pairs $(\ms M,\mc E)\in{\bf Obs}_d\times{\bf Ch}_d$. We treat any $\Gamma\in{\bf Ins}_d$ as a sequence $(\Gamma_j)_{j\in\mb Z_d}$ of operations summing up to a trace-preserving operation.

For any $q,\,p\in\mb Z_d$, we define the operators $U_q,\,V_p,\,W_{q,p}\in\mc L(\hil)$ in the same way as in Equations (\ref{eq:Usiirto})-(\ref{eq:Weyl}) with the basis $\{\f_j\}_{j\in\mb Z_d}$ replaced with the basis $\{|n\ra\}_{n\in\mb Z_d}$, so that, e.g.,\ $U_q|n\ra=|n+q\ra$ for all $q,\,n\in\mb Z_d$. We denote the set of those $\ms M\in{\bf Obs}_d$ such that $W_{q,p}^*\ms M_jW_{q,p}=\ms M_{j-q}$ for all $q,\,p,\,j\in\mb Z_d$ by ${\bf Obs}_d^{W,1}$. Furthermore, we denote by ${\bf Ch}_d^{W,2}$ the set of those $\mc E\in{\bf Ch}_d$ such that $\mc E(W_{q,p}\rho W_{q,p}^*)=W_{q,p}\mc E(\rho)W_{q,p}^*$ for all $q,\,p\in\mb Z_d$ and all $\rho\in\mc S(\hil)$. Finally, we denote the set of those $\Gamma\in{\bf Ins}_d$ such that $\Gamma_{j-q}(W_{q,p}\rho W_{q,p}^*)=W_{q,p}\Gamma_j(\rho)W_{q,p}^*$ for all $j,\,q,\,p\in\mb Z_d$ and all $\rho\in\mc S(\hil)$ by ${\bf Ins}_d^W$.

The elements in the sets ${\bf Obs}_d^{W,1}$, ${\bf Ch}_d^{W,2}$, and ${\bf Ins}_d^W$ have a simple structure: For any $\ms M\in{\bf Obs}_d^{W,1}$, there is a positive operator $C\in\mc L(\hil)$ such that $V_pC=CV_p$ for all $p\in\mb Z_d$, $\sum_jU_jCU_j^*=\id_\hil$ (so that, especially, $\tr{C}=1$), and
\begin{equation}\label{eq:Wkovsuure}
\ms M_j=U_jCU_j^*,\qquad j\in\mb Z_d.
\end{equation}
For each $\Gamma\in{\bf Ins}_d^W$, there is an operation $\mc D:\mc L(\hil)\to\mc L(\hil)$ such that $\mc D(V_p\rho V_p^*)=V_p\mc D(\rho)V_p^*$ for all $p\in\mb Z_d$ and $\rho\in\mc S(\hil)$, and $\sum_jU_j\mc D^*(\id_\hil)U_j^*=\id_\hil$, and
\begin{equation}\label{eq:Wkovins}
\Gamma_j(\rho)=U_j\mc D(U_j^*\rho U_j)U_j^*,\qquad j\in\mb Z_d,\quad\rho\in\mc S(\hil).
\end{equation}
For the characterizations given above for $W$-covariant observables and instruments, see \cite[Section III]{Holevo1998}; the conjecture presented in the reference certainly holds in our discrete case. For any covariant channel $\mc E\in{\bf Ch}_d^{W,2}$ there is a positive kernel $(q,p)\mapsto\Phi_{q,p}\in\mb C$ such that \cite{HoWe01}
\begin{equation}\label{eq:Wker}
\mc E^*(W_{q,p})=\Phi_{q,p}W_{q,p},\qquad q,\,p\in\mb Z_d.
\end{equation}
Since the operators $W_{q,p}$, $q,\,p\in\mb Z_d$, span the whole of $\mc L(\hil)$, the kernel $\Phi$ completely characterizes the covariant channel. Positivity of the kernel $\Phi$ means that the Fourier transform of the kernel is positive, i.e.,\ for any $j,\,k\in\mb Z_d$,
\begin{equation}\label{eq:posit}
\hat\Phi_{j,k}:=\frac{1}{d}\sum_{q,p\in\mb Z_d}\ovl{\dual{q}{k}}\dual{j}{p}\Phi_{q,p}\geq0.
\end{equation}
Moreover, when $\mc E$ is defined as in (\ref{eq:Wker}), then
$$
\mc E(\rho)=\frac{1}{d}\sum_{j,k\in\mb Z_d}\hat\Phi_{j,k}W_{j,k}\rho W_{j,k}^*,\qquad \rho\in\mc S(\hil).
$$
When the channel $\mc E$ arises from a covariant instrument like that in Equation (\ref{eq:Wkovins}), it follows from straight-forward calculation (utilizing the fact that, for any $B\in\mc L(\hil)$, one has $\sum_{q,p}W_{q,p}BW_{q,p}^*=d\tr{B}\id$) that the kernel $\Phi$ associated with $\mc E$ is of the form
\begin{equation}\label{eq:kerneli}
\Phi_{q,p}=\tr{W_{q,p}^*\mc D^*(W_{q,p})},\qquad q,\,p\in\mb Z_d.
\end{equation}

As earlier, we have the covariantization maps ${\bf Obs}_d\ni\ms M\mapsto\ms M^{W,1}\in{\bf Obs}_d^{W,1}$, ${\bf Ch}_d\ni\mc E\mapsto\mc E^{W,2}\in{\bf Ch}_d^{W,2}$, and ${\bf Ins}_d\ni\Gamma\mapsto\Gamma^W\in{\bf Ins}_d^W$ having the covariant devices as their fixed points. Especially,
$$
\Gamma_j^W(\rho)=\frac{1}{d^2}\sum_{q,p\in\mb Z_d}W_{q,p}^*\Gamma_{j-q}(W_{q,p}\rho W_{q,p}^*)W_{q,p}
$$
for any $\Gamma\in{\bf Ins}_d$, $j\in\mb Z_d$, and $\rho\in\mc S(\hil)$. Moreover, by similar arguments as earlier, one can show that, if a pair $(\ms M,\mc E)\in{\bf Obs}_d\times{\bf Ch}_d$ is compatible, so is $(\ms M^{W,1},\mc E^{W,2})$, and if a pair $(\ms M,\mc E)\in{\bf Obs}_d^{W,1}\times{\bf Ch}_d^{W,2}$ is compatible, it has a joint instrument in ${\bf Ins}_d^W$. As in preceding analyses, it follows:

\begin{lemma}\label{lemma:Wkov}
Suppose that $(\ms A,\mc A)\in{\bf Obs}_d^{W,1}\times{\bf Ch}_d^{W,2}$. One has
$$
W(\ms A,\mc A)=\sup\{w(\ms A,\mc A|\ms B,\mc B)\,|\,\ms B\in{\bf Obs}_d^{W,1},\ \mc B\in{\bf Ch}_d^{W,2}\}.
$$
\end{lemma}

In what follows, we study the robustness of incompatibility of the pair $(\ms A,\mr{id})$, where $\ms A_n=|n\ra\la n|$ for all $n\in\mb Z_d$. Evidently, $\ms A\in{\bf Obs}_d^{W,1}$ and $\mr{id}\in{\bf Ch}_d^{W,2}$, but this pair is not compatible. According to the preceding lemma, there is a pair $(\ms B,\mc B)\in{\bf Obs}_d^{W,1}\times{\bf Ch}_d^{W,2}$ such that the pair
$$
\big(W\ms A+(1-W)\ms B, W\mr{id}+(1-W)\mc B\big),
$$
with $W:=W(\ms A,\mr{id})$, is compatible, and hence has a joint instrument in ${\bf Ins}_d^W$. We start determining the value of $W(\ms A,\mr{id})$ first by giving a characterization for the instruments of ${\bf Ins}_d^W$.

\begin{lemma}\label{lemma:Wkovmatr}
For any $\Gamma\in{\bf Ins}_d^W$ there is an indexed set $\alpha=(\alpha_{r,s}^n)_{n,r,s\in\mb Z_d}$ of complex numbers such that, for all $n\in\mb Z_d$, the matrix $(\alpha_{r,s}^n)_{r,s\in\mb Z_d}$ is positive, $\sum_{n,r}\alpha^n_{r,r}=1$, and defining $\mc D:\mc L(\hil)\to\mc L(\hil)$,
\begin{equation}\label{eq:matrop}
\mc D(A)=\sum_{n,r,s\in\mb Z_d}\alpha^n_{r,s}\la n+r|A|n+s\ra|r\ra\la s|,\qquad A\in\mc L(\hil),
\end{equation}
$\Gamma$ is defined as in Equation (\ref{eq:Wkovins}).
\end{lemma}

\begin{proof}
Suppose that $\Gamma\in{\bf Ins}_d^W$ is associated with an operation $\mc D$ that is covariant with respect to the representation $p\mapsto V_p$ (i.e.,\ is $V$-covariant), $\tr{\mc D^*(\id)}=1$, and $\Gamma$ is defined as in (\ref{eq:Wkovins}). The $V$-covariance means that, for the Choi operator $M=(\mc D^*\otimes\mr{id})(\Om_d)$, where $\Om_d=\sum_{m,n}|m,m\ra\la n,n|$, we have
$$
(V_p\otimes\ovl V_p)M=M(V_p\otimes\ovl V_p),\qquad p\in\mb Z_d,
$$
where $\ovl V_p=V_p^{*T}$ with the transpose defined with respect to the basis $\{|n\ra\}_n$. Since the representation $p\mapsto V_p\otimes\ovl V_p$ is associated with the spectral measure $q\mapsto\ms R_q$,
$$
\ms R_q=\sum_{r\in\mb Z_d}|q+r,r\ra\la q+r,r|,\qquad q\in\mb Z_d,
$$
i.e.,\ $V_p\otimes\ovl V_p=\sum_{q}\dual{q}{p}\ms R_q$, it follows that any operator that commutes with this representation has to be a linear combination of the operators $A_{n,r,s}=|n+r,r\ra\la n+s,s|$, $n,\,r,\,s\in\mb Z_d$. Writing $M=\sum_{n,r,s}\alpha_{r,s}^nA_{n,r,s}$, one finds that $\mc D$ is defined as in (\ref{eq:matrop}). The positivity of $M$ is equivalent with $\ms R_nM\ms R_n\geq0$ for all $n$ which, in turn, is equivalent with $(\alpha_{r,s}^n)_{r,s}\geq0$ for all $n$. Direct calculation shows that $\sum_jU_j\mc D^*(\id)U_j^*=\id$ is equivalent with $\sum_{n,r}\alpha_{r,r}^n=1$.
\end{proof}

Let us pick a compatible pair $(\ms M,\mc E)\in{\bf Obs}_d^{W,1}\times{\bf Ch}_d^{W,2}$ that has a joint instrument $\Gamma\in{\bf Ins}_d^W$ that is associated with $\alpha$ like that in Lemma \ref{lemma:Wkovmatr}. Through simple calculations one finds that the trace-1 positive operator $C$ associated with $\ms M$ according to (\ref{eq:Wkovsuure}) and the positive kernel $\Phi$ associated with $\mc E$ according to (\ref{eq:Wker}) are given by
\begin{equation}\label{eq:alphakaavat}
C=\sum_{n,r\in\mb Z_d}\alpha_{r,r}^n|n+r\ra\la n+r|,\qquad\Phi_{q,p}=\sum_{n,r\in\mb Z_d}\ovl{\dual{n}{p}}\alpha_{r,r-q}^n,\quad q,\,p\in\mb Z_d.
\end{equation}
In the case of our special incompatible pair $(\ms A,\mr{id})$, the trace-1 positive operator associated with $\ms A$ is $|0\ra\la0|$ and the positive kernel associated with the identity channel is the constant kernel $1$. It follows that, if we require the compatible pair like that above to be a convex combination of the form
$$
(\ms M,\mc E)=t(\ms A,\mr{id})+(1-t)(\ms B,\mc B)
$$
with some $t\in[0,1]$ and some (not necessarily compatible) pair $(\ms B,\mc B)\in{\bf Obs}_d\times{\bf Ch}_d$, then we must have $D=C-t|0\ra\la0|\geq0$ and the kernel $\Psi=\Phi-t$ has to be positive, where $C$ and $\Phi$ are defined as in (\ref{eq:alphakaavat}). Since $C$ is diagonalized in the basis $\{|n\ra\}_n$, it is a straight-forward check that the first condition is equivalent with
\begin{equation}\label{eq:Dpos}
t\leq\sum_{n\in\mb Z_d}\alpha_{-n,-n}^n=:w_1(\alpha).
\end{equation}
Using the characterization of (\ref{eq:posit}) for the positivity of the kernel $\Psi$, the latter condition is equivalent with
\begin{equation}\label{eq:Psipos}
t\leq\frac{1}{d}\sum_{r,s\in\mb Z_d}\alpha_{r,s}^0=:w_2(\alpha).
\end{equation}
The robustness $W(\ms A,\mr{id})$ is simply the supremum of $\min\{w_1(\alpha),\,w_2(\alpha)\}$ over all those $\alpha=(\alpha_{r,s}^n)_{n,r,s}$ such that $(\alpha_{r,s}^n)_{r,s}\geq0$ for all $n\in\mb Z_d$ and $\sum_{n,r}\alpha_{r,r}^n=1$.

We may simplify the optimization task presented above by a couple of observations: First, we note that, given $\alpha$, the elements $\alpha_{r,s}^n$ where $n\neq0$, $r\neq-n$, and $s\neq-n$ can be assumed to be zero; this assumption does not affect the value of $w_2$ and it can only increase the value of $w_1$. The elements $\alpha_{-n,-n}^n$ are, of course, non-negative. Second, the property $\sum_{n,r}\alpha_{r,r}^n=1$ and the positivity requirements are not violated if we replace the $\alpha$, where the elements $\alpha_{r,s}^n$ with $n\neq0$, $r\neq-n$, and $s\neq-n$ are zero, with $\tilde\alpha$, where $\tilde\alpha_{0,0}^0=\sum_{n\in\mb Z_d}\alpha_{-n,-n}^n$, $\tilde\alpha_{r,s}^n=0$ for all $n\neq0$ and $r,\,s\in\mb Z_d$, and the rest of the entries in $\tilde\alpha$ coincide with their counterparts in $\alpha$. In this process, the value of $w_1$ does not change but $w_2$ may increase. Hence, in our optimization, it suffices to study only those $\alpha$ such that $\alpha_{r,s}^n=0$ for any $n\neq0$ and $r,\,s\in\mb Z_d$ and, for the upper block $(\alpha_{r,s}^0)_{r,s}=:A$, $A\geq0$ and $\tr{A}=1$. From now on, we forget about the zero blocks and only concentrate on the upper block $A$. Let us denote the natural basis of $\mb C^d$ by $\{e_r\}_{r\in\mb Z_d}$ and define the Fourier operator $\mc F\in\mc L(\mb C^d)$ as in (\ref{eq:fourier}) with the basis $\{\f_j\}_j$ replaced by $\{e_n\}_n$. Fix the basis $\{f_n\}_n$, $f_n=\mc Fe_n$, $n\in\mb Z_d$. We may define $w_1(A)=\sis{e_0}{Ae_0}$, $w_2(A)=\sis{f_0}{Af_0}$, and $w_0(A)=\min\{w_1(A),\,w_2(A)\}$. We have found that $W(\ms A,\mr{id})$ is the supremum of $w_0(A)$ over the positive trace-1 operators $A\in\mc L(\mb C^d)$.

This optimization task can still be further simplified: For any $A\in\mc L(\mb C^d)$, let us define $A^{\mc F}\in\mc L(\mb C^d)$ through
$$
A^{\mc F}=\frac14\sum_{k=1}^4\mc F^kA\mc F^{*k}.
$$
The fixed points of the map $A\mapsto A^{\mc F}$ are exactly the Fourier-invariant operators $B$, i.e.,\ $\mc FB=B\mc F$. One finds that $w_0(A^{\mc F})=w_1(A^{\mc F})=w_2(A^{\mc F})=\frac12\big(w_1(A)+w_2(A)\big)\geq w_0(A)$ for all positive $A\in\mc L(\mb C^d)$. Thus, our task is simply to optimize the linear functional $w_0(A)=\sis{e_0}{Ae_0}$ over the set of positive trace-1 Fourier-invariant operators $A$ on $\mb C^d$. The optimal value is reached at one of the extreme points of the set of such operators, and these extreme points coincide with the projections onto the one-dimensional subspaces generated by the eigenvectors of $\mc F$. The Fourier operator has four eigenvalues, the fourth roots of 1 $i^k$, $k=1,\,2,\,3,\,4$, and the corresponding eigenprojections are
$$
P_k=\frac14\big(\id+(-i)^k\mc F+(-1)^k\mc F^2+ i^k\mc F^3\big),\quad k=1,\,2,\,3,\,4.
$$
Hence, especially, for an extreme point $A$ of the set of trace-1 positive Fourier-invariant operators, there is a unique $k\in\{1,\,2,\,3,\,4\}$ such that $A=P_kAP_k$ and $w_0(A)=\sis{P_ke_0}{AP_ke_0}$. One finds that $P_1e_0=P_3e_0=0$ and $P_2e_0=\frac12(e_0-f_0)$ and $P_4e_0=\frac12(e_0+f_0)$ which means that $w$ is maximized at a projection onto a one-dimensional subspace of the $P_2$- or $P_4$-eigenspace. It is obvious that, amongst such operators supported on the $P_2$-eigenspace, the optimal one is $A_-=|v_-\ra\la v_-|$ and, amongst the extreme points supported on the $P_4$-eigenspace, the highest value for $w_0$ is given by $A_+=|v_+\ra\la v_+|$, where
$$
v_\pm=\sqrt{\frac{\sqrt{d}}{2(\sqrt{d}\pm1)}}(e_0\pm f_0).
$$
Simple check shows that $w_0(A_-)<w_0(A_+)=\frac12(1+1/\sqrt{d})$. Thus, $W(\ms A,\mr{id})=\frac12(1+1/\sqrt{d})$. Again, for a decodable channel $\mc V\in{\bf Ch}(\hil,\mc K)$, one has $(\ms A,\mc V)=_\mr{post}(\ms A,\mr{id})$, so that, according to Theorem \ref{theor:prop3}:

\begin{theorem}\label{theor:Aidjoin}
The robustness of incompatibility for any von Neumann observable $\ms A\in{\bf Obs}_d$ and any decodable channel $\mc V\in{\bf Ch}(\hil,\mc K)$ is
$$
W(\ms A,\mc V)=\frac12\Big(1+\frac{1}{\sqrt{d}}\Big).
$$
\end{theorem}

The optimal positive trace-1 Fourier-invariant operator $A_+$ of the discussion preceding Theorem \ref{theor:Aidjoin} gives rise to an operation $\mc D$ through Equation (\ref{eq:matrop}) which, in turn, defines an instrument $\Gamma$ according to (\ref{eq:Wkovins}). It follows that
$$
\Gamma_j(\rho)=\frac{\sqrt{d}}{2(\sqrt{d}+1)}\Big(\frac{1}{\sqrt{d}}\id+\ms A_j\Big)\rho\Big(\frac{1}{\sqrt{d}}\id+\ms A_j\Big).
$$
The optimal decomposition for the compatible pair $(\ms M,\mc E)=(\Gamma_{(1)},\Gamma_{(2)})$ is given by
$$
(\ms M,\mc E)=\frac12\Big(1+\frac{1}{\sqrt{d}}\Big)(\ms A,\mr{id})+\frac12\Big(1-\frac{1}{\sqrt{d}}\Big)(\ms B,\mc B),
$$
where
$$
\ms B=-\frac{1}{d-1}\ms A+\frac{d}{d-1}\ms T,\qquad\mc B=-\frac{1}{d-1}\mr{id}+\frac{d}{d-1}\mc E_{\ms A},
$$
where $\ms T$ is the trivial observable $\ms T_j=\frac{1}{d}\id$, $j\in\mb Z_d$, and $\mc E_{\ms A}\in{\bf Ch}_d^{W,2}$ is the L\"uders channel $\rho\mapsto\sum_j\ms A_j\rho\ms A_j$ associated with $\ms A$. Moreover, the optimal compatible pair $(\ms M,\mc E)$ is given by
$$
\ms M=\frac{\sqrt{d}+2}{2(\sqrt{d}+1)}\ms A+\frac{\sqrt{d}}{2(\sqrt{d}+1)}\ms T,\qquad\mc E=\frac{\sqrt{d}+2}{2(\sqrt{d}+1)}\mr{id}+\frac{\sqrt{d}}{2(\sqrt{d}+1)}\mc E_{\ms A}.
$$
Similarly, for a decodable channel $\mc V$, the pair
$$
\frac12\Big(1+\frac{1}{\sqrt{d}}\Big)(\ms A,\mc V)+\frac12\Big(1-\frac{1}{\sqrt{d}}\Big)(\ms B_{\mc V},\mc B_{\mc V})=\frac{\sqrt{d}+2}{2(\sqrt{d}+1)}(\ms A,\mc V)+\frac{\sqrt{d}}{2(\sqrt{d}+1)}(\ms T,\mc V\circ\mc E_{\ms A})
$$
is compatible, where
$$
\ms B_{\mc V}=-\frac{1}{d-1}\ms A+\frac{d}{d-1}\ms T,\qquad\mc B_{\mc V}=-\frac{1}{d-1}\mc V+\frac{d}{d-1}\mc V\circ\mc E_{\ms A},
$$

\section{Conclusions}

Given a convex subset $L_0$ of a real vector space, we have introduced measures of how well an element of the minimal affine subspace $F$ containing $L_0$ resists immersion in $L_0$ under added noise. Especially, we have concentrated on the case where $L_0\subset K\subset F$ and $K$ is convex, in which case such measures can be defined relative to $K$. As a physical application, these robustness measures were studied in the case where $K$ is the set of pairs of given quantum devices (observables or channels) and $L_0$ is the set of compatible pairs within $K$. In this context, we call such a measure as robustness of incompatibility.

Basic properties of the robustness measures have been investigated especially regarding monotonicity under certain compatibility-increasing partial orderings of device pairs. Lastly, values for the absolute robustness were calculated in three exemplary cases: a pair of Fourier-coupled rank-1 sharp observables, a pair of decodable channels (especially for unitary channels), and a pair consisting of a rank-1 sharp observable and a decodable channel.

However, we do not have a general method for how to calculate the robustness of incompatibility for a general pair of quantum devices; all our examples utilize symmetries in calculating the values of the robustness measure. Moreover, especially in the case of compatibility of observables and channels, it would, perhaps, be more natural to consider the other robustness function $w$ instead of $W$; any noise in a measurement process affects both the registering branch and the state-change branch globally and is hence, typically, compatible. However, calculating the value for $w$ in the example involving a von Neumann observable and a decodable channel leads to quite a complicated optimization problem. Moreover, calculating the robustness of incompatibility for the infinite-dimensional sharp Weyl-pair (position-momentum pair in $L^2(\mb R)$) is a possible continuation of the analysis dealing with the finite-dimensional case of Section \ref{sec:weyl}. For the time being, we only conjecture this infinite-dimensional pair to have the robustness $1/2$, i.e.,\ this pair would be an example of a maximally incompatible pair according to the robustness measure.

\subsection*{Acknowledgements}

The author would like to thank Dr. Teiko Heinosaari, Dr. David Reeb, Dr. Jussi Schultz, and Dr. Michal Sedl\'ak  for inspiring discussions and suggestions as well as the esteemed referees for their constructive feedback. Financial support from the Doctoral Programme in Physical and Chemical Sciences (PCS) of the University of Turku is acknowledged.

\end{document}